\documentclass[12pt]{article}
\usepackage{amsmath}
\usepackage{graphicx,psfrag,epsf}
\usepackage{enumerate}
\usepackage{fullpage}
\usepackage{natbib}
\usepackage{url} 

\newcommand{\blind}{1}

\usepackage{amssymb,amsthm,bm,color,tcolorbox,caption,subcaption}

\usepackage{tikz}
\usetikzlibrary{shadings,intersections,decorations.markings}
\usepackage{pgfplots}
\pgfplotsset{width=7cm,compat=1.8}

\newtheorem{assumption}{Model Assumption}
\newtheorem{proposition}{Proposition}
\newtheorem{corollary}{Corollary}

\begin{document}


\if1\blind
{
  \title{Bayesian Probabilistic Numerical Methods in Time-Dependent State Estimation for Industrial Hydrocyclone Equipment}
  \author{Chris. J. Oates$^{1,2}$, Jon Cockayne$^3$, Robert G. Aykroyd$^4$, Mark Girolami$^{5,2}$ \\
$^1$Newcastle University \\
$^2$Alan Turing Institute \\
$^3$University of Warwick \\
$^4$University of Leeds \\
$^5$Imperial College London }
  \maketitle
} \fi

\if0\blind
{
  \bigskip
  \bigskip
  \bigskip
  \begin{center}
    {\LARGE\bf Bayesian Probabilistic Numerical Methods in Time-Dependent State Estimation for Industrial Hydrocyclone Equipment}
\end{center}
  \medskip
} \fi

\begin{abstract}
The use of high-power industrial equipment, such as large-scale mixing equipment or a hydrocyclone for separation of particles in liquid suspension, demands careful monitoring to ensure correct operation.
The fundamental task of state-estimation for the liquid suspension can be posed as a time-evolving inverse problem and solved with Bayesian statistical methods.
In this paper, we extend Bayesian methods to incorporate statistical models for the error that is incurred in the numerical solution of the physical governing equations.
This enables full uncertainty quantification within a principled computation-precision trade-off, in contrast to the over-confident inferences that are obtained when all sources of numerical error are ignored.
The method is cast within a sequential Monte Carlo framework and an optimised implementation is provided in \verb+Python+.
\end{abstract}

\noindent%


\section{Introduction}

Hydrocyclones provide a simple and inexpensive method for removing solids from liquids, as well as separating two liquids according to their relative densities (assuming equal fluid resistances) \citep{Gutierrez2000}.
They have widespread applications, including in areas such as environmental engineering and the petrochemical industry \citep{Sripriya2007}.
In particular, they have few moving parts, can handle large volumes and are relatively inexpensive to maintain.
This makes them ideal as part of a continuous process in hazardous industrial settings and contrasts with alternatives, such as filters and centrifuges, which are more susceptible to breakdown and/or have higher running costs.
The physical principles governing the hydrocyclone are simple; a mixed input is forced into a cone-shaped tank at high pressure, to create a circular rotation.
This rotation forces less-dense material to the centre and denser material to the periphery of the tank. 
The less-dense material in the core can then be extracted from the top (overflow) and the denser material removed from the bottom of the tank (underflow).
This mechanism is illustrated in Fig. \ref{fig:hydrocyclone sketch}.

Continual monitoring of the hydrocyclone is essential in most industrial applications, since the input flow rate is an important control parameter that can be adjusted to maximise the separation efficiency of the equipment.
Our focus in this work is on state estimation for the internal fluid.
Indeed, the high pressures that are often involved necessitate careful observation of the internal fluid dynamics to ensure safety in operation \citep{Bradley2013}.

\begin{figure}
\centering

\begin{subfigure}[b]{0.45\textwidth}
\centering
\begin{tikzpicture}
  \coordinate (O) at (0,-1);
  
  \begin{axis}[
    view       = {-25}{-25},
    axis lines = middle,
    axis line style={draw=none},
    zmax       = 30,
    height     = 5cm,
    xtick      = \empty,
    ytick      = \empty,
    ztick      = \empty,
    width=110pt,
    at={(-30pt,-37pt)}
  ]
  \addplot3+ [
    ytick      = \empty,
    yticklabel = \empty,
    domain     = 0:8*pi,
    samples    = 400,
    samples y  = 0,
    mark       = none,
    black,
  ]
  ( {x*sin(0.28*pi*deg(x))},{x*cos(0.28*pi*deg(x)},{x});
  \end{axis}

  \begin{scope}
    \def\rx{0.31}
    \def\ry{0.15}
    \def\z{1.45}

    \path [name path = ellipse]    (0,\z) ellipse ({\rx} and {\ry});
    \path [name path = horizontal] (-\rx,\z-\ry*\ry/\z) -- (\rx,\z-\ry*\ry/\z);
    \path [name intersections = {of = ellipse and horizontal}];

    \draw[fill = gray!20, blue!10] (intersection-1) -- (0,0.5)
      -- (intersection-2) -- cycle;
    
    \draw[fill = gray!20, densely dashed] (0,\z) ellipse ({\rx} and {\ry});
  \end{scope}

  \draw (0.25,0.4) -- (0.9,0.1) node at (1.8,0.0) {more dense};
  \draw (0,0.9) -- (-0.9,0.1) node at (-1.8,0.0) {less dense};

  \filldraw (O) circle (1pt) node[below] {underflow};

  \draw[] (O) to (-1.33,1.33);
  \draw[] (O) -- (1.33,1.33);

  \draw[black] (-1.36,1.46) arc [start angle = 170, end angle = 10,
    x radius = 13.8mm, y radius = 3.6mm];
  \draw[black] (-1.29,1.52) arc [start angle=-200, end angle = 20,
    x radius = 13.75mm, y radius = 3.15mm];

  \draw[ball color=blue,shading=ball, opacity = 0.2] (1.37,1.42) arc [start angle = 0, end angle = 360,
    x radius = 13.7mm, y radius = 3.3mm];

  \draw (-1.2,2.2) -- (-0.1,1.5) node at (-1.37,2.37) {overflow};
  
\end{tikzpicture}
\caption{Hydrocyclone tank schematic}
\end{subfigure}
\begin{subfigure}[b]{0.45\textwidth}
\centering
\begin{tikzpicture}

\draw[ball color=blue,shading=ball, opacity = 0.2,line width=4pt] (2,0) to [out = 180, in = 0] (0,0) to [out = 180, in = 90] (-1,-1) to [out = -90, in = 180] (0,-2) to [out = 0, in = -90] (1,-1) to [out = 90, in = -45] (0.7071,-0.2929) to [out = 0, in = 180] (2,-0.2929);

\node (n1) at (2,-0.15) {};
\node (n2) at (3.5,-0.15) {input flow};
\path[->] (n2) edge (n1);

\node (n3) at (-0.5,-1) {};
\node (n4) at (0,-1.5) {};
\path[->] (n3) edge [bend right = 45] (n4);

\node (n5) at (0,-1.5) {};
\node (n6) at (0.5,-1) {};
\path[->] (n5) edge [bend right = 45] (n6);

\node (n7) at (0.5,-1) {};
\node (n8) at (0,-0.5) {};
\path[->] (n7) edge [bend right = 45] (n8);

\node (n9) at (0,-0.5) {};
\node (n10) at (-0.5,-1) {};
\path[->] (n9) edge [bend right = 45] (n10);

\end{tikzpicture} \vspace{20pt}
\caption{Cross-section (top of tank)}
\end{subfigure}
\caption{A simplified schematic description of typical hydrocyclone equipment.
(a) The tank is cone-shaped with overflow and underflow pipes positioned to extract the separated contents.
(b) Fluid, a mixture to be separated, is injected at high pressure at the top of the tank to create a vortex.
Under correct operation, less-dense materials are directed toward the centre of the tank and denser materials are forced to the peripheries of the tank.}
\label{fig:hydrocyclone sketch}
\end{figure}
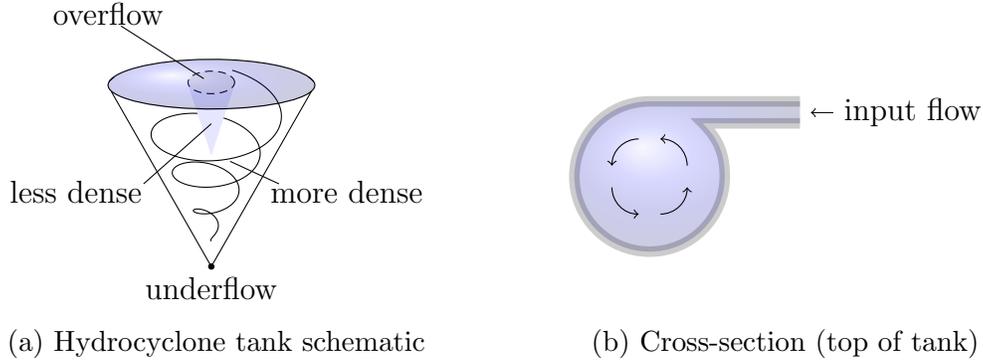

\subsection{Statistical Challenges}

Direct observation of the internal flow of the fluids is difficult or impossible due to, for example, the reinforced walls of the hydrocyclone and the opacity of the mixed component.
Under correct operation, the output (overflow and underflow) can be measured and tested for purity, but advanced indication of a potential loss of efficiency is desirable, if not essential in most industrial contexts.
Such a warning allows for adjustment and hence avoidance of impending catastrophic failure.
One possible technique for monitoring the internal flow is electrical impedance tomography (EIT).
The target of an EIT analysis is the electrical conductivity field $a^\dagger$ of the physical object; the conductivities of different fluid components will in general differ and this provides a means to measure the fluid constituents.
This technique has many applications in medicine, as well as industry, as it provides a non-invasive method to estimate internal structure from external measurements (the \emph{inverse problem}) \citep{Gutierrez2000}.
Further, it is ideal for industrial processes as it is possible to collect data at rates of several hundred frames per second, hence allowing real-time monitoring and control of sensitive systems.
However, the rapid acquisition of data requires equally rapid analysis and the nature of EIT requires that low-accuracy approximations to the physical governing equations are needed to keep pace with incoming data in the monitoring context \citep{Hamilton2018}.
This is due to the computational demands that are posed by the repeated solution of physical governing equations (the \emph{forward problem}) in evaluation of the statistical likelihood.
However, in standard approaches, the error introduced by a crude discretisation of the physical governing equations is not accounted for and may lead to an over-optimistic view of the precision of results.
This could lead to misleading interpretations of the results and hence potentially dangerous mis-control; it is therefore important to account for the presence of an unknown and non-negligible discretisation error in interpretation of the statistical output.

\subsection{Probabilistic Numerical Methods}

Probabilistic numerics \citep{Hennig2015} is an emergent research field that aims to model the uncertainty in the solution space of the physical equations that arises when the forward problem is only approximately solved.
In contrast to conventional emulation methods \citep{Kennedy2001}, which are {\it extrusive} in the sense that the physical equations are treated as a black box, probabilistic numerical methods are {\it intrusive} and seek to model the error introduced in the numerical solution due to discretisation of the original continuous physical equations.
Thus a probabilistic numerical method provides \emph{uncertainty quantification for the forward problem} that is meaningful, reflective of the specific discretisation scheme employed, and enables a principled computation-precision trade-off, where the presence of an unknown discretisation error is explicitly accounted for by marginalisation over the unknown solution to the forward problem \citep{Briol2016,Cockayne2017}.
This paper contributes a rigorous assessment of probabilistic numerical methods for the Bayesian solution of an important inverse problem in industrial process monitoring, detailed next.

\subsection{Our Contributions}

The scientific problem that we consider is Bayesian state estimation for the time-evolving conductivity field of internal fluid using data obtained via EIT.
The Bayesian approach to inverse problems is well-studied \citep{Stuart2010,Nouy2014} and in particular the application of statistical methods to EIT is now well-understood \citep{Kaipio1999,Kaipio2000,Watzenig2009,Dunlop2015,Yan2015,Aykroyd2015,Stuart2016} with sophisticated computational methods proposed \citep{Kaipio2000,Vauhkonen2001,Polydorides2002,Schwab2012,Schillings2013,Beskos2015,Chen2015,Hyvoonen2015,Chen2016a,Chen2016b,Chen2016}.
In this paper, probabilistic numerical methods are proposed and investigated as a natural approach to uncertainty quantification with a computation-precision trade-off, wherein numerical error in the approximate solution of the forward problem is explicitly modelled and accounted for in a full Bayesian solution to the inverse problem of interest.
At present, the literature on probabilistic numerical methods for partial differential equations consists of \cite{Owhadi2015a,Chkrebtii2016,Owhadi2016,Owhadi2016a,Cockayne2016,Cockayne2016b,Conrad2015,Raissi2017}.
This paper goes further than the most relevant work in \cite{Cockayne2016,Cockayne2016b}, which tackled Bayesian inverse problems based on EIT with probabilistic numerical methods, in several aspects:
\begin{itemize}
\item The inversion problem herein is more challenging than the (static) problems considered in previous work, in that we aim to recover the temporal evolution of the true unknown conductivity field $a^\dagger$ based on indirect and noise-corrupted observations at a finite set of measurement times.
To address this challenge, a (descriptive, rather than mechanistic) Markovian prior model $a$ for the field is developed, which is shown to admit a filtering formulation \citep{Todescato2017}.
This permits a sequential Monte Carlo method (particle filter) to be exploited for efficient data assimilation \citep{Law2015}.
\item The filtering formulation introduces additional challenges due to the fact that numerical (discretisation) error in solution of the forward problem will be propagated through computations performed at earlier time points to later time points, as well as the possibility that numerical errors can accumulate within the computations.
A detailed empirical investigation of the computation-precision trade-off is undertaken based on the use of probabilistic numerical methods for solution of the EIT governing equations.
\item Real experimental data are analysed, generated by one of the present authors in a controlled laboratory experiment.
These data consist of 2,401 individual voltage measurements taken at discrete spatial and temporal intervals over the boundary of the vessel, and are used to demonstrate the efficacy of the approach under realistic experimental conditions.
\end{itemize}
In particular, this paper constitutes one of the first serious applications of probabilistic numerical methods to a challenging real-world problem, where proper quantification of uncertainty is crucial.

\subsection{Overview of the Paper}

The structure of the paper is as follows:
Sec. \ref{sec:methods} contains the mathematical, statistical and computational methodological development.
Sec. \ref{sec:experiments} reports our experimental results and Sec. \ref{sec:discussion} discusses their implications for further research and for future application to industrial processes.

\section{Methods} \label{sec:methods}

In Sec. \ref{subsec:abstraction} we introduce the physical model and make the inversion problem formal.
Then in Sec. \ref{subsec:Bayesian} we recall the Bayesian approach to inversion, with an extension to a time-evolving unknown parameter.
Sec. \ref{subsec:PNM} introduces probabilistic models for numerical error incurred in discretisation of the physical governing equations.
The final section, \ref{subsec:SMC} develops a sequential Monte Carlo method for efficient computation.

\subsection{Abstraction of the Inverse Problem} \label{subsec:abstraction}

The physical equations that model the measurement process are presented below, following the recent comprehensive treatment in \cite{Dunlop2015}.

\subsubsection{Set-Up}

Consider a bounded, open domain $D \subset \mathbb{R}^d$ with smooth boundary denoted $\partial D$.
Let $\bar{D} = D \cup \partial D$.
The domain represents a physical object and our parameter of interest is the conductivity field $a : \bar{D} \rightarrow \mathbb{R}$ of that object.
Here $a(x)$ denotes the conductivity at spatial location $x \in \bar{D}$.
Consider $m$ electrodes fixed to $\partial D$, the region of contact of electrode $i \in \{1,\dots,m\}$ being denoted $E_i \subset \partial D$.
A current stimulation pattern $\mathrm{I} = (I_i)_{i=1}^m \in \mathbb{R}^m$ is passed, via the electrodes, through the object.
Note that from physical conservation of current we have $\sum_{i=1}^m I_i = 0$.

The electrical potential $u : \bar{D} \rightarrow \mathbb{R}$ over the domain, induced by the current stimulation pattern, can be described by the following partial differential equation (PDE):
\begin{eqnarray}
\arraycolsep=1.4pt\def\arraystretch{1.5}
\begin{array}{rcll}
\nabla \cdot (a \nabla u) & = & 0 & \text{in } D \\
\int_{E_i} a \nabla u \cdot \mathrm{n} \mathrm{d} \sigma & = & I_i &  \\
u & = & U_i & \text{on } E_i  \\
a \nabla u \cdot \mathrm{n} & = & 0 & \text{on } \partial D \setminus \cup_{i=1}^m E_i . 
\end{array} \label{PDE}
\end{eqnarray}
Here $\mathrm{n}$ is the \emph{outward} unit normal, which corresponds to the convention that $I_i > 0$ refers to current flow \emph{out} of the domain, and $\mathrm{d}\sigma$ represents an infinitesimal boundary element.
The quantities $U_i$ on the electrodes $E_i$ will constitute the measurements.
Known as the \emph{complete electrode model} (CEM), this PDE\footnote{The mathematical formulation in Eqn. \ref{PDE} assumes, as we do in this work, that contact impedance at the electrodes can be neglected. For the case of imperfect electrodes, the reader is referred to \cite{Aykroyd2018}.} was first studied in \cite{Cheng1989}.
For a suitable fixed field $a$, existence of a solution $u$ is guaranteed and, under the additional condition that $\sum_{i=1}^m U_i = 0$, uniqueness of the solution $u$ can also be established \citep{Somersalo1992}.
Thus the forward problem is well-defined.

The true conductivity field $a^\dagger$ is considered to be unknown and is the object of interest.
In contrast to most work on EIT, in our context $a^\dagger$ is time-dependent and we extend the notation as $a^\dagger(x,t)$ for, with no loss in generality, a time index $t \in [0,1]$.
In order to estimate $a^\dagger$, measurements $\mathrm{y}_{j,k}$ are obtained under distinct stimulation patterns $\mathrm{I}_j \in \mathbb{R}^m$, $j = 1,\dots,J$, modelled as
\begin{eqnarray} 
\mathrm{y}_{j,k} & = & \mathcal{P}_{j,k} u^\dagger + \epsilon_{j,k} \in \mathbb{R}^m \label{measurement model}
\end{eqnarray}
where the projections
\begin{eqnarray*} 
\mathcal{P}_{j,k} u^\dagger & := & \left[ \begin{array}{c} u(x_1^E ; \mathrm{I}_j , a^\dagger(\cdot,t_k) ) \\ \vdots \\ u(x_m^E ; \mathrm{I}_j , a^\dagger(\cdot,t_k) ) \end{array} \right]
\end{eqnarray*}
are defined for each stimulation pattern $j$ and each discrete time point $t_k$, $k = 1,\dots,n$, and the $\epsilon_{j,k}$ represent error in the measurement.
Here $u( \cdot ; \mathrm{I} , a )$ denotes the solution of the PDE with conductivity field $a$ and stimulation pattern $\mathrm{I}$, while $x_i^E$ is a point central to the electrode $E_i$.
Thus $u^\dagger = u(\cdot ; \mathrm{I} , a^\dagger)$ is the solution of the PDE defined by the true field $a^\dagger$ and, for fixed $j,k$, the vector $\mathcal{P}_{j,k} u^\dagger$ contains the quantities $U_i$ in the CEM with conductivity field $a^\dagger(\cdot,t_k)$ and stimulation pattern $\mathrm{I} = \mathrm{I}_j$.

In the absence of further conditions on $a^\dagger$, the inverse problem is ill-posed.
Indeed, the infinite-dimensional field $a^\dagger$ cannot be uniquely recovered from a finite dataset.
(Recall the seminal work of \cite{Hadamard1902}, who defined an inverse problem to be \emph{well-posed} if (i) a solution exists, (ii) the solution is unique, and (iii) the solution varies continuously as the data are varied.)
To proceed, the inverse problem must therefore be regularised \citep{Tikhonov1977}.

\subsection{The Bayesian Approach to Inversion} \label{subsec:Bayesian}

In this section we exploit Bayesian methods to regularise the inverse problem \citep{Stuart2010}.
Sec. \ref{prior model setup} introduces the prior model, Sec. \ref{filtering sec} casts posterior computation as a filtering problem and Sec. \ref{error bounds overview} reviews mathematical analysis for numerical approximation of the posterior that accounts for numerical error in the PDE solution method.

\subsubsection{Prior Model for the Conductivity Field} \label{prior model setup}

In this paper we interpret Eqn. \ref{PDE} in the strong form, which in particular requires the existence of $\nabla a^\dagger$ on $D$.
This information will be encoded into a prior distribution:
Let $\{\phi_i\}_{i=1}^\infty$ be an orthonormal basis for a separable Hilbert space $H$ with norm $\|\cdot\|_{H}$.
It is assumed that $H \subset C^1(\bar{D})$, where $C^m(S)$ is used to denote the set of $m$-times continuously differentiable functions from $S$ to $\mathbb{R}$.
Let $(\Omega,\mathcal{F},\mathbb{P})$ be a probability space and for measurable $v : \Omega \rightarrow \mathbb{R}$ denote $\mathbb{E} v = \int v \mathrm{d} \mathbb{P}$.

\begin{assumption} \label{prior model for a}
Let $\alpha > 1/2$ and $\omega \in \Omega$.
Our prior model is expressed as a separable Karhounen-Lo\'{e}ve expansion:
\begin{eqnarray*}
\log \; a(x,t; \omega) & = & \sum_{i=1}^\infty i^{- \alpha} \psi_i(t ; \omega) \phi_i(x) 
\end{eqnarray*}
where the $\omega \mapsto \psi_i(\cdot ; \omega)$ are modelled as independent Gaussian processes with mean functions $m_{\psi,i}$ and covariance functions $k_{\psi,i}$ such that 
\begin{eqnarray*}
m_\psi^{\max} \; := \; \sup_{i \in \mathbb{N}} \sup_{t \in [0,1]} |m_{\psi,i}(t)| < \infty, \hspace{30pt}
k_\psi^{\max} \; := \; \sup_{i \in \mathbb{N}} \sup_{t \in [0,1]} k_{\psi,i} (t,t) < \infty .
\end{eqnarray*}
\end{assumption}
The logarithm is used to ensure positivity of the conductivity field, as is considered standard in Bayesian approaches to (static) EIT \citep{Dunlop2015}.
Henceforth the probability argument $\omega \in \Omega$ will be left implicit.

This prior construction ensures that $\nabla a$ exists in $D$.
To see this, we have the following result:
\begin{proposition} \label{prop: is cts}
For fixed $t \in [0,1]$, almost surely $a(\cdot,t)$ exists in $C^1(\bar{D})$.
\end{proposition}
Note that, in particular, this result justifies point evaluation of $\nabla a$ in the algorithms that we present; since $x \mapsto \nabla a(x,t)$ is almost surely continuous, such point evaluations are almost surely well-defined.
All proofs are reserved for Appendix \ref{app: proofs}.

This paper imparts weak prior assumptions, in the sense of smoothness, on the time-evolution of the random field:

\begin{assumption}
The $\psi_i$ are modelled as Brownian with mean functions $m_{\psi,i}(t) = 0$ and covariance functions $k_{\psi,i}(t,t') = \lambda \min(t + \tau,t' + \tau)$, for all $t, t' \in [0,1]$, for some fixed $\lambda > 0$ and $\tau \geq 0$.
\label{Markov assumption}
\end{assumption}
This prior model allows for flexible and data-driven estimation of the temporal evolution of the unknown conductivity field.
At the same time, this choice allows estimation to be cast as a filtering problem (see Sec. \ref{filtering sec}) due to the following important fact:

\begin{proposition}[Due to \cite{Wiener1949}] \label{field increments}
The increments $\psi_i(t+s) - \psi_i(t)$ are independent with distribution $\mathrm{N}(0,\lambda (s + \tau))$, for all $0 \leq t \leq t + s \leq 1$.
\end{proposition}

An immediate consequence is that the field $a(\cdot,t)$ itself is a Markov process.
To see this, let
\begin{eqnarray*}
k_\phi(x,x') = \sum_{i=1}^\infty i^{-2\alpha} \phi_i(x) \phi_i(x') .
\end{eqnarray*}
For convenience, we let $\theta = \log a$ in the sequel. 
Then:

\begin{corollary} \label{field increment corr}
The increments $\theta_\Delta(\cdot) := \theta(\cdot,t+s) - \theta(\cdot,t)$ are independent Gaussian random fields with mean function $m_\Delta(x) = 0$ and covariance function $k_\Delta(x,x') =  \lambda (s + \tau) k_\phi(x,x')$, for all $x,x' \in D$ and all $0 \leq t \leq t + s \leq 1$. 
\end{corollary}

Let $\Gamma_s$ denote the distribution of the increment $\theta_\Delta$ over the time interval $[t,t+s]$.
The infinite-dimensional nature of the random variable $\theta_\Delta$ precludes the use of standard density notation, due to the non-existence of a Lebesgue measure in the infinite-dimensional context \citep[p.143][]{Yamasaki1985}. 
Instead, the distribution $\Gamma_s$ is formalised through its Radon-Nikodym derivative  
\begin{eqnarray*}
\frac{\mathrm{d}\Gamma_s}{\mathrm{d}\gamma} (\theta_\Delta) & \propto & \exp\left( -\frac{1}{2} \|\theta_\Delta\|_{k_\Delta}^2 \right) \; = \; \exp\left( -\frac{1}{2 \lambda (s + \tau)} \|\theta_\Delta\|_{k_\phi}^2 \right)
\end{eqnarray*}
with respect to abstract Wiener measure $\gamma$ \citep{Gross1967}, where $\|\cdot\|_k$ denotes the Cameron-Martin norm based on the covariance function $k$.
The reader unfamiliar with Radon-Nikodym notation is referred to the accessible introduction in \cite{Halmos1949}.

\subsubsection{Formulation as a Filtering Problem} \label{filtering sec}

Denote $\theta_k = \theta(\cdot,t_k)$.
Then the directed acyclic graph representation of the conditional independence structure of the statistical model \citep{Lauritzen1996} is as follows:
\begin{eqnarray*}
\begin{array}{ccccccccc}
\theta_1 & \rightarrow & \theta_2 & \rightarrow & \dots & \rightarrow & \theta_{n-1} & \rightarrow & \theta_n \\
\downarrow & & \downarrow & & & & \downarrow & & \downarrow \\
\mathrm{y}_{\cdot,1} & & \mathrm{y}_{\cdot,2} & & \dots & & \mathrm{y}_{\cdot,n-1} & & \mathrm{y}_{\cdot,n}
\end{array}
\end{eqnarray*}
Let $\Pi_n$ represent the posterior distribution over the conductivity field based on the data $\mathrm{y}_{\cdot,k}$ for $k \leq n$.
Then statistical inference is naturally formulated as a filtering problem at linear cost \citep{Saerkkae2013,Todescato2017}:
\begin{eqnarray*}
\frac{\mathrm{d}\Pi_n}{\mathrm{d}\Pi_{n-1}}(\theta) & \propto & p(y_{\cdot,n} | \mathrm{I} , \theta_n ).
\end{eqnarray*}
Here the Radon-Nikodym notation has been used on the LHS, while $p$ has the conventional interpretation as a p.d.f. with respect to Lebesgue measure, here representing the likelihood model specified by the distributional model for $\epsilon_{j,k}$ in Eqn. \ref{measurement model}.
The solution to the filtering problem is the $n$-step posterior distribution:
\begin{eqnarray*}
\frac{\mathrm{d}\Pi_n}{\mathrm{d}\Pi_{0}}(\theta) & = & \prod_{k=1}^n \frac{\mathrm{d}\Pi_k}{\mathrm{d}\Pi_{k-1}}(\theta)
\end{eqnarray*}
where the reference measure $\Pi_0$ is the prior distribution for the conductivity field given in Assumption \ref{prior model for a}.
From Prop. \ref{field increments}, the prior marginal on $(\theta_1,\dots,\theta_n)$, denoted $\Pi_{0,1:n}$, can be decomposed as follows:
\begin{eqnarray*}
\frac{\mathrm{d}\Pi_{0,1:n}}{\mathrm{d}(\gamma \times \dots \times \gamma)}(\theta_1,\dots,\theta_n) & \propto & \frac{\mathrm{d}\Pi_{0,1}}{\mathrm{d}\gamma}(\theta_1) \prod_{k=2}^n \frac{\mathrm{d}\Gamma_{t_k - t_{k-1}}}{\mathrm{d}\gamma}(\theta_k - \theta_{k-1}) ,
\end{eqnarray*}
where $\gamma \times \dots \times \gamma$ denotes the product of $n$ abstract Wiener measures and the initial distribution $\Pi_{0,1}$ is computed as
\begin{eqnarray*}
\frac{\mathrm{d}\Pi_{0,1}}{\mathrm{d}\gamma} (\theta_1) & \propto & \exp\left( - \frac{1}{2 \lambda (t_1 + \tau)} \|\theta_1\|_{k_\phi}^2 \right) 
\end{eqnarray*}
as a direct consequence of Assumption \ref{Markov assumption}.
Later we use $\Pi_{n,n+1}$ to denote the marginal of $\Pi_n$ over the components $\theta(\cdot,t_{n+1})$; the so-called \emph{posterior predictive} distribution.
This is simply a convolution of $\Pi_n$ with the centred Gaussian field described in Cor. \ref{field increment corr} and is of industrial relevance since it allows anticipation of the future dynamics and thus for intelligent hazard control.

\subsubsection{Numerical Error and its Analysis} \label{error bounds overview}

The likelihood model in Eqn. \ref{measurement model} depends on the projections $\mathcal{P}_{j,k} u$ which in turn depend on the exact solution $u(x_i^E ; \mathrm{I}_j , a(\cdot,t_k) )$ of the PDE for given inputs $\mathrm{I}_j$ and $a(\cdot,t_k)$.
In general the exact solution of the PDE is unavailable in closed-form and numerical methods are used to obtain discrete approximations, for instance based on a finite element or collocation basis \citep{Quarteroni2008}. 
The assessment of the error introduced through discretisation is well-studied, with sophisticated theories for worst-case and average-case errors and beyond \citep{Novak2008,Novak2010b}.

Several papers have leveraged these analyses to consider the impact of discretisation error in the forward problem on the inferences that are made for the inverse problem
\citep{Schwab2012,Schillings2013,Schillings2014,Nouy2014,Bui-Thanh2014,Chen2015,Chen2016,Chen2016a,Nagel2016}.
These analyses all focus on static inverse problems (i.e. for a single time point).
However, the generalisation of these theoretical results to the temporal context introduces considerable technical difficulties.
Indeed, the filtering formulation is such that error in an approximation of $\Pi_1$ will be propagated and lead to an error in the approximation of $\Pi_n$ whenever $n \geq 2$. 
Numerical approximation of $\Pi_n$ thus involves $n-1$ sources of discretisation error and analysis in the time-evolving setting must account for propagation and accumulation of these discretisation errors.
However, standard worst-case error analyses (such as those listed above) are inappropriate for temporal problems, since in general the worst-case scenario will not be realised simultaneously by all numerical methods involved in the computational work-flow.

Presented with such an inherently challenging problem, our novel approach - described in the next section - to \emph{model} discretisation error as an unknown random variable and propagate \emph{uncertainty due to discretisation} through computation has appeal on philosophical, technical and practical levels.

\subsection{Probabilistic Numerical Methods} \label{subsec:PNM}

Recall that the exact solution $u(\cdot ; a , \mathrm{I})$ to the PDE is unavailable in closed-form.
In this section we view numerical solution of Eqn. \ref{PDE} not as a forward problem, but as an inverse problem in its own right (called a \emph{sub-inverse} problem in this work) and provide full quantification of solution uncertainty that arises from the discretisation of this PDE via a collocation-type method.
Sec. \ref{prior for PDE} introduces a prior model for $u$ while Sec. \ref{PMM sec} completes the specification of this sub-inverse problem associated with solution of the PDE.
Then, Sec. \ref{marginal like sec} demonstrates how solution uncertainty can be propagated through the original inverse problem by marginalisation over the unknown exact solution $u$ of the PDE.
Sec. \ref{theory sec} establishes theoretical properties of the proposed method.

\subsubsection{Prior Model for the Potential Field} \label{prior for PDE}

In this section we again adopt Bayesian methods to make the sub-inverse problem well-posed.
The chief task is to construct a prior for $u$, the potential field.
In principle, the physical governing equations, together with the prior for the conductivity field $a$, induce a unique prior for the potential field.
The relationship between these probabilities has been explored in the context of stochastic PDEs; see \cite{Lord2014} for a book-length treatment.
However, the task of characterising (or even approximating) the implied distribution on $u$ is highly non-trivial\footnote{In principle this is characterised by the Green's function of the PDE, but if the Green's function was known we would not have needed to discretise the PDE.}.
For this reason, we follow \cite{Cockayne2016,Cockayne2016b} and treat the two unknown fields as independent under the prior model.
In particular, we encode independence across time points into the prior model for $u$, a choice that is algorithmically convenient.
This allows us, in the following, to leave the time index implicit.
This has a natural statistical interpretation of encoding only partial information into the prior - and can be both statistically and pragmatically justified \citep{Potter1983}.

To reduce notation in this and the following section, we consider a fixed conductivity field $a \in C^1(\bar{D})$ and a fixed current stimulation pattern $\mathrm{I} \in \mathbb{R}^m$; these will each be left implicit.

\begin{assumption}
\label{two derivatives}
The unknown solution $u$ to Eqn. \ref{PDE} is modelled as a Gaussian process with mean function $m_u(x) = 0$ and covariance function 
\begin{eqnarray}
k_u(x,x') & = & \int k_u^0(x,z) k_u^0(z,x') \mathrm{d}z \label{int kernel}
\end{eqnarray}
such that $k_u^0 \in C^{2 \times 2}(\bar{D} \times \bar{D})$ is a positive-definite kernel.
\end{assumption}

This minimal assumption ensures that, \emph{under the prior}, the differential $\nabla \cdot( a \nabla u)$ is well-defined over $D$.
Indeed, in general:
\begin{proposition} \label{prior well defined for u}
If $k_u^0 \in C^{\beta \times \beta}(\bar{D} \times \bar{D})$ with $\beta \in \mathbb{N}$, then almost surely $u \in C^\beta(\bar{D})$. 
\end{proposition}

\subsubsection{Probabilistic Meshless Method} \label{PMM sec}

Next we obtain a posterior distribution over the solution $u$ to the PDE in Eqn. \ref{PDE}.
In particular this requires us to be explicit about the nature of our ``data'' for this sub-inverse problem.
The mathematical justification for our approach below is provided in the information-based complexity literature on linear elliptic PDEs of the form $Au = f$ on $D$, $Bu = g$ on $\partial D$ \citep{Werschulz1996,Novak2008,Cialenco2012}.
In this framework, limited data $\mathrm{f}_i = f(x_i^A)$, $\mathrm{g}_i = g(x_i^B)$ are provided on the forcing term $f$ and the boundary term $g$; the mathematical problem is then optimal recovery of the solution $u$ from these data, under a loss function that must be specified.
This is a particular example of a \emph{linear information} problem, since the $\mathrm{f}_i$ and $\mathrm{g}_i$ are linear projections of the unknown solution $u$ of interest; see \cite{Novak2008} for a book length treatment.

The data with which we work, in the above sense, are linear projections obtained at \emph{collocation} points $\{x_i^{A}\}_{i=1}^{n_A} \subset D$ and $\{x_i^{B}\}_{i=1}^{n_B} \subset \partial D$:
\begin{eqnarray*}
\begin{array}{rclcll}
\mathcal{L}_i u & := & \nabla \cdot a(x_i^{A}) \nabla u(x_i^{A}) & = & 0 & i = 1,\dots,n_A \\[8pt]
\mathcal{L}_{n_A + i} u & := & a(x_i^{B}) \nabla u(x_i^{B}) \cdot \mathrm{n}(x_i^{B}) & = & 0 & i = 1,\dots,n_B \\[8pt]
\mathcal{L}_{n_A+n_B+i} u & := & \int_{E_i} a \nabla u \cdot \mathrm{n} \mathrm{d}\sigma & = & I_i & i = 1,\dots,m.
\end{array}
\end{eqnarray*}
Here $\mathcal{L} = [\mathcal{L}_1,\dots,\mathcal{L}_{n_A+n_B+m}]$ is a linear operator from $C^2(\bar{D})$ to $\mathbb{R}^{n_A+n_B+m}$.
For a function $h(\cdot , \cdot) \in C^{2 \times 2}(\bar{D} \times \bar{D})$, in a slight abuse of notation, $\mathcal{L}h$ will be used to denote action of $\mathcal{L}$ on the first argument, while the notation $\bar{\mathcal{L}}h$ denotes action on the second argument.
The composition $\mathcal{L} \bar{\mathcal{L}}h$ is understood as a matrix with $(i,j)$th element $\mathcal{L}_i \bar{\mathcal{L}}_j h \in \mathbb{R}$.
In this notation, the data can be expressed as $\mathcal{L} u = [0^\top , \mathrm{I}^\top]^\top$ where $\mathrm{I} = [I_1,\dots,I_m]^\top$.
The posterior over $u$ is obtained by conditioning the prior measure on these data.
Recall that $\mathcal{P} u = [u(x_1^E) , \dots , u(x_m^E)]^\top$.
For our purposes, it is sufficient to obtain the posterior over the finite dimensional vector $\mathcal{P}u$:
\begin{eqnarray}
\mathcal{P} u \left| \mathcal{L}u = \left[ \begin{array}{c} 0 \\ \mathrm{I} \end{array} \right] \right. & \sim & \mathrm{N}( \mu , \Sigma ) \nonumber \\
\mu & = & [\mathcal{P} \bar{\mathcal{L}} k_u] [\mathcal{L}\bar{\mathcal{L}} k_u]^{-1} \left[ \begin{array}{c} 0 \\ \mathrm{I} \end{array} \right] \label{PMM eqns} \\
\Sigma & = & [\mathcal{P}\bar{\mathcal{P}} k_u] - [\mathcal{P} \bar{\mathcal{L}} k_u] [\mathcal{L} \bar{\mathcal{L}} k_u]^{-1} [\mathcal{L} \bar{\mathcal{P}} k_u] \nonumber
\end{eqnarray}
This distribution represents uncertainty due to the finite amount of computation that is afforded to numerical solution of the PDE in Eqn. \ref{PDE}.
Eqn. \ref{PMM eqns} was termed a \emph{probabilistic meshless method} in \cite{Cockayne2016,Cockayne2016b}.
Note that the \emph{maximum a posteriori} estimate $\mu$ is identical to the point estimate provided by symmetric collocation \citep{Fasshauer1996} and this point estimator (only) was considered in the context of Bayesian PDE-constrained inverse problems in \cite{Marzouk2009a,Yan2015}.
The point estimator $\mu$ requires that the $(n_A + n_B + m)$-dimensional square matrix $\mathcal{L} \bar{\mathcal{L}} k_u$ is inverted; since this is also the computational bottleneck in computation of $\Sigma$, it follows that the probabilistic meshless method has essentially the same computational cost as its non-probabilistic counterpart. 
Considerable theoretical advances in the numerical analysis of these probabilistic numerical methods (for static problems) have since been made in \cite{Owhadi2016}.
For non-degenerate kernels $k_u$, the matrix $\mathcal{L} \bar{\mathcal{L}} k_u$ is of full rank provided that no two collocation points are coincidental.

The selection of collocation points can be formulated as a problem of statistical experimental design.
Indeed, adaptive refinement strategies, that target an appropriate functional of the posterior covariance $\Sigma$ until a pre-specified tolerance is met, can be considered \citep[see][]{Cockayne2016}.
For brevity in this paper we simply considered the collocation points to be fixed.

\subsubsection{Marginal Likelihood} \label{marginal like sec}

The natural approach to define a data distribution is through marginalisation over the unknown solution $u$ to the PDE.
This marginalisation can be performed in closed form under a Gaussian measurement error model:
\begin{assumption}
The measurement errors $\epsilon_{j,k}$ are independent $\mathrm{N}(0,\sigma^2 I)$.
\label{independent measurements}
\end{assumption}
Consider a stimulation pattern $\mathrm{I}_j$ applied at time $t_k \in [0,1]$.
Define $P_{j,k} = \mathcal{P}u( \cdot ; \mathrm{I}_j , a(\cdot,t_k) )$ and denote by $\mu_j , \Sigma_j$ the output of the probabilistic meshless method (Eqn. \ref{PMM eqns}) for the input stimulation pattern $\mathrm{I}_j$.
Then the marginal distribution of the data $\mathrm{y}_{j,k}$, given the measurement error standard deviation $\sigma$, admits a density as follows:
\begin{eqnarray}
p^*(\mathrm{y}_{j,k} | \mathrm{I}_j , a(\cdot , t_k) , \sigma ) & = & \int \mathrm{N}( \mathrm{y}_{j,k} | P_{j,k} , \sigma^2 I ) \mathrm{N}( P_{j,k} | \mu_j , \Sigma_j ) \mathrm{d}P_{j,k} \nonumber \\
& = & \mathrm{N}( \mathrm{y}_{j,k} | \mu_j ,\sigma^2 I + \Sigma_j ) \label{PMM likelihood}
\end{eqnarray}
where we have used the shorthand of $\mathrm{N}(\cdot | \mu_j , \Sigma_j)$ for the p.d.f. of $\mathrm{N}(\mu_j , \Sigma_j)$.
Eqn. \ref{PMM likelihood} has the clear interpretation of inflating the measurement error covariance $\sigma^2 I$ by an additional amount $\Sigma_j$ to reflect additional uncertainty due to discretisation error in the numerical solution of the PDE in Eqn. \ref{PDE}.
This distinguishes the probabilistic approach from other applications of collocation methods in the solution of Bayesian PDE-constrained inverse problems, where uncertainty due to discretisation is ignored \citep{Marzouk2009a,Yan2015}.
Eqn. \ref{PMM likelihood} also appears in the emulation literature for static problems \citep[e.g.][]{Calvetti2017}.
However, emulation methods treat the PDE as a perfect black-box and, as a result, the matrices $\Sigma_j$ obtained from emulation do not reflect the fact that the PDE must be discretised\footnote{The typical usage of emulators is to reduce the total number of forward problems that must be solved.
This consideration is orthogonal to the present work and the two approaches could be combined.}.

This paper proposes to base statistical inferences on the posterior distribution $\Pi_n^*$ defined recursively via
\begin{eqnarray*}
\frac{\mathrm{d}\Pi_n^*}{\mathrm{d}\Pi_{n-1}^*}(\theta) & \propto & p^*(\mathrm{y}_{\cdot,n} | \mathrm{I} , \exp(\theta_n) , \sigma) , \qquad \Pi_0^* = \Pi_0 .
\end{eqnarray*}
In particular we will be most interested in the posterior predictive distribution $\Pi_{n,n+1}^*$ obtained with these probabilistic numerical methods, where discretisation uncertainty is explicitly modelled.
Unlike $\Pi_{n,n+1}$, the posterior predictive distribution $\Pi_{n,n+1}^*$ can be exactly computed, since it does not require the exact solution of the PDE.

\subsection{Theoretical Properties} \label{theory sec}

The theoretical analysis of \cite{Cockayne2016} can be exploited to assess the consistency of the probabilistic meshless method in Eqn. \ref{PMM eqns}, in the case where the field $a$ is fixed.
Define the \emph{fill distance} $h := \min\{h_A,h_B\}$ where $h_A = \sup_{x \in D} \min_i \|x - x_i^A \|_2$ and $h_B = \sup_{x \in \partial D} \min_i \|x - x_i^B\|_2$.
Then we outline the following result:

\begin{proposition} \label{fill dist result}
Let $B_\epsilon$ denote a Euclidean ball of radius $\epsilon > 0$ centred on $\mathcal{P}u$ in $\mathbb{R}^m$, where $u$ is the true solution of the PDE and $\mathcal{P}$ was as previously defined.
Then, under the assumptions of \cite{Cockayne2016}, which include that $H(k_u^0)$ is norm-equivalent to the Sobolev space $H^\beta(D)$, then the mass afforded to $\mathbb{R}^m \setminus B_\epsilon$ in the posterior $\mathrm{N}(\mu_j , \Sigma_j)$ is $O(\epsilon^{-2} h^{2\beta - 4 - d})$ for $h > 0$ sufficiently small.
\end{proposition}

This result ensures asymptotic agreement between the probabilistic numerical approach to the inverse problem and the (unavailable) exact approach based on the exact solution of the PDE in Eqn. \ref{PDE} in the limit $h \rightarrow 0$ of infinite computation.
Empirical evidence for the appropriateness of the uncertainty quantification for static EIT experiments and finite computation was presented in \cite{Cockayne2016}.

\subsection{Computation via Sequential Monte Carlo} \label{subsec:SMC}

The log-normal prior on the conductivity field precludes a closed-form posterior.
However, the filtering formulation of Sec. \ref{filtering sec} suggests a natural approach to computation based on particle filters, otherwise known as sequential Monte Carlo (SMC) methods \citep{DelMoral2004}.

Let $\Pi_{0,1} \ll \Pi_{0,1}'$ where $\Pi_{0,1}'$ is a user-chosen importance distribution on $C^1(\bar{D})$ (and could be $\Pi_0$).
The method begins with $N$ independent draws $\theta_0^{(1)},\dots,\theta_0^{(N)}$ from $\Pi_{0,1}'$; each draw $\theta_0^{(i)}$ is associated with an importance weight 
$$
w_0^{(i)} \propto \frac{\mathrm{d}\Pi_{0,1}}{\mathrm{d}\Pi_{0,1}'}(\theta_0^{(i)})
$$
such that $\sum_{i=1}^N w_0^{(i)} = 1$.
This provides an empirical approximation $\sum_{i=1}^N w_0^{(i)} \delta(\theta_0^{(i)})$ to the prior marginal distribution $\Pi_{0,1}$ that becomes exact as $N$ is increased.
Let $t_0 := t_1$.
Then, at each iteration $n = 1,2,\dots$ of the SMC algorithm, the following steps are performed:
\begin{enumerate}
\item {\bf Re-sample:} Particles $\tilde{\theta}_n^{(1)},\dots,\tilde{\theta}_n^{(N)}$ are generated as a random sample (with replacement) of size $N$ from the empirical distribution $\sum_{i=1}^N w_{n-1}^{(i)} \delta(\theta_{n-1}^{(i)})$.
\item {\bf Move:} Each particle $\tilde{\theta}_n^{(i)}$ is updated to $\theta_n^{(i)}$ according to a Markov transition $M_{n-1}$ that leaves $\Pi_{n-1}^*$ invariant. 
(Details are provided in Appendix \ref{ap: markov kernels}.)
\item {\bf Re-weight:} The next set of weights are defined as 
$$
w_n^{(i)} \propto p^*(\mathrm{y}_{\cdot,n} | \mathrm{I} , \exp(\theta_n^{(i)}) , \sigma)
\frac{\mathrm{d}\Gamma_{t_n - t_{n-1}}}{\mathrm{d}\gamma}(\theta_n^{(i)} - \tilde{\theta}_n^{(i)})
$$
and such that $\sum_{i=1}^N w_n^{(i)} = 1$.
\end{enumerate}
The output after $n$ iterations is an empirical approximation $\sum_{i=1}^N w_n^{(i)} \delta(\theta_n^{(i)})$ to the posterior distribution $\Pi_n^*$.
The posterior predictive distribution $\Pi_{n,n+1}^*$ can be obtained from similar methods, as
\begin{eqnarray}
\frac{\mathrm{d} \Pi_{n,n+1}^*}{\mathrm{d} \Pi_n^*}(\theta) & \propto & \frac{\mathrm{d} \Gamma_{t_{n+1} - t_n} }{\mathrm{d} \gamma}(\theta_{n+1} - \theta_n) . \label{eq: predictive posterior}
\end{eqnarray}
The re-sample step in the above procedure does not in general need to occur at each iteration, only when the effective sample size is small; see \cite{DelMoral2004}.
Theoretical analysis of SMC methods in the context of infinite-dimensional state spaces is provided in \cite{Beskos2015}.
For this work we considered a fairly standard SMC method, but several extensions are possible and include, in particular, stratified or quasi Monte Carlo re-sampling methods \citep{Gerber2015}.
One extension which we explored was to introduce fictitious intermediate distributions between $\Pi_{n-1}^*$ and $\Pi_n^*$ following \cite{Chopin2002}, which we found to improve the performance of SMC in this context.
For the experiments reported in the paper, 100 intermediate distributions were used, defined by tempering on a linear temperature ladder, c.f. \cite{Kantas2014,Beskos2015}.

This completes our methodological development.
Optimised \verb+Python+ scripts are available to reproduce these results at:
\if1\blind
{\url{https://github.com/jcockayne/hydrocyclone_code}.
} \fi
\if0\blind
{\url{blinded_url}
} \fi
Next, we report empirical results based on data from a controlled EIT experiment.

\section{Results} \label{sec:experiments}

This section considers data from a laboratory experiment designed to investigate the temporal mixing of two liquids.
The experiment was conducted by one of the present authors and carefully controlled, to enable assessment of statistical methods and to mimic the salient features of industrial hydrocyclone equipments.

\subsection{Experimental Protocol}

In the experiment, a cylindrical perspex tank of diameter 15cm and height 30cm was used with a single ring of $m=8$ electrodes, each measuring approximately 1cm wide by 3cm high.
The electrodes start at the bottom of the tank, with the initial liquid level exactly at the top of the electrodes.
Hence there is translation invariance in the vertical direction and the contents are effectively a single 2D region, meaning that electrical conductivity can be modelled as a 2D field.
The experimental set-up is depicted in Fig. \ref{fig:beaker}.

At the start of the experiment, a mixing impeller was used to create a rotational flow.
This was then removed and, after a few seconds, concentrated potassium chloride solution was carefully injected into the tap water initially filling the tank.
Data was then collected at regular time intervals until it was assumed that the liquid had fully mixed.
Further details of the experiment can be found in \cite{West2005}.
These data were previously analysed (with non-probabilistic numerical methods) in \cite{Aykroyd2007}. 

This experiment mimics the situation when a hydrocyclone moves from an in-control regime to an out-of-control regime, in that initially there is a well defined core which gradually disappears as the liquids merge together.
Performing the experiment in the laboratory allowed careful control of experimental conditions and, in particular, a lack of electrical interference from other equipment.
A similar experimental set-up for data-generation was recently employed in \cite{Hyvoonen2015}.

There are several widely accepted data collection `protocols' for EIT \citep{Isaacson1986}.
A protocol specifies the sequence of electrodes that are used to create the electric field, as well as the sequence of electrodes used to measure the resulting electric potential.
In this experiment the `reference protocol' was used, where a drive current is passed between a reference electrode and each of the other electrodes in turn allowing a maximum of $J = 7$ linearly independent current patterns.
For each current pattern, the $U_1,\dots,U_m$ were measured up to a common additive constant\footnote{This reflects the fact that it is voltage that is actually measured, which is the difference of two potentials.}, so that without loss of generality $E_1$ is the `reference' electrode and $U_1 \equiv 0$.
This permits a total of $7\times 7=49$ measurements $\mathrm{y}_{\cdot,k}$, obtained at each time point $t_k$ in the experiment.

\begin{figure}[t!]
\centering
\includegraphics[width = 0.3\textwidth]{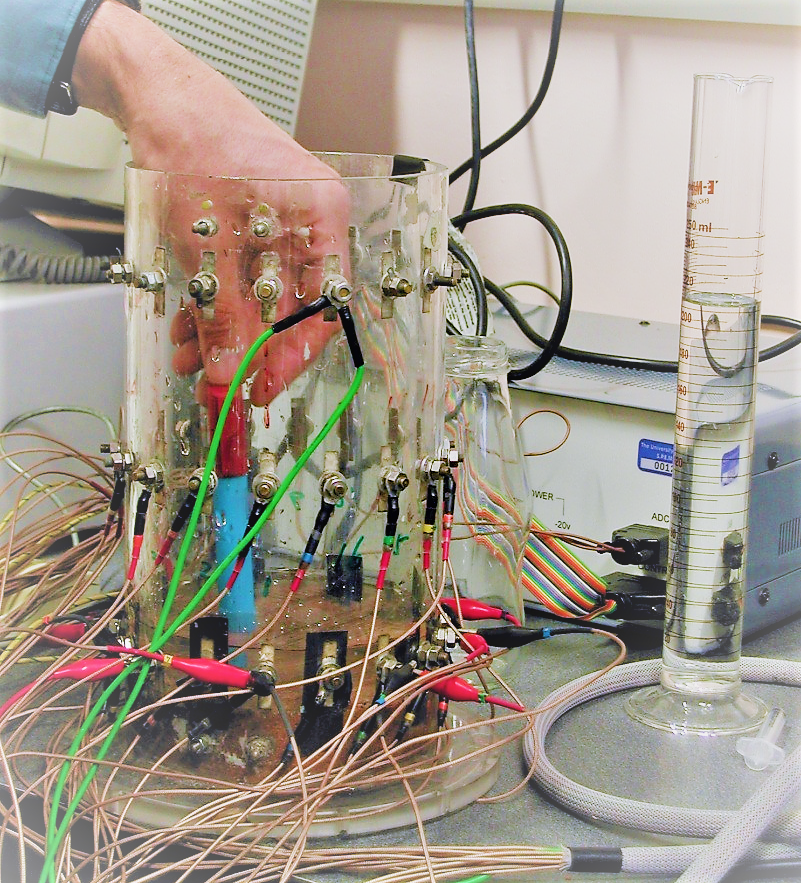}
\caption{Experimental set-up: A cylindrical perspex tank containing tap water was stirred before an amount of potassium chloride was injected.
Electrodes positioned around the tank measured voltages, which can be related through a partial differential equation to the internal conductivity field.
The inverse problem consists of estimating the internal conductivity field from the voltages that were measured.
(Only the bottom ring of electrodes were used for the data analysed in this paper.)
Photo reproduced from \cite{West2005}.}
\label{fig:beaker}
\end{figure}

\subsection{Experimental Results}

The proposed statistical approach, based on probabilistic numerical methods, was used to make inferences on the unknown conductivity field $a^\dagger$ based on this realistic experimental dataset.

The assumption that $k_u^0$ has two continuous derivatives is sufficient for the prior to be well-defined (Prop. \ref{prior well defined for u}).
However, the theoretical result in Prop. \ref{fill dist result} requires a more regular kernel with at least $\beta > 2 + d/2$ (weak) derivatives to ensure contraction of the (static) posterior.
In reality, molecular diffusion implies that clear boundaries are not expected to be present in the true conductivity field. 
Thus it is reasonable to assume that both the conductivity field $a$ and the electrical potential $u$ will be fairly smooth in the interior $D$.
For these reasons, the kernels employed for experiments below were of squared-exponential form, since this trivially meets all smoothness requirements, including smoothness of the solution $u$ in $D$.

\subsubsection{Static Recovery Problem}

First, we calibrated our probabilistic numerical methods by analysing the static recovery problem.
This prior for $\theta$ was taken to be Gaussian, with a mean of zero and a squared-exponential covariance function
\begin{equation*}
  k_a(\bm{x}, \bm{x}') := \varphi_a \exp\left( 
   -\frac{\|\bm{x}-\bm{x}'\|_2^2}{2 \ell_a^2}
  \right)
\end{equation*}
where $\varphi_a$ controls the magnitude of fields drawn from the prior, while the length-scale $\ell_a$ controls how rapidly those functions vary. 
Since the main aim here is to assess the probabilistic meshless method, rather than the performance at state estimation, we simply fixed $\varphi_a = 1$ and $\ell_a = 0.3$.
Note that while it is common in EIT problems to use priors which promote hard edges in drawn samples, owing to applications in medicine, here a smooth prior is appropriate.
For all experiments in this paper the parameter $\sigma$, that describes technical measurement error, was set to $\sigma = 1.0$ based on analysis of a technical replicate dataset.
For the probabilistic meshless method, the prior model was centered and a squared-exponential covariance function was used, with $\varphi_u=100$ to match the scale of measurements in the dataset, and $\ell_u=0.211$, a value chosen by empirical Bayes based upon a high-quality reference sample.
The collocation points were chosen on concentric circles, as shown in Fig. \ref{fig:design} for increasing values of $n_A$ and $n_B$.

\begin{figure}[t]
\centering
\begin{subfigure}{0.3\textwidth}
\includegraphics[width = \textwidth]{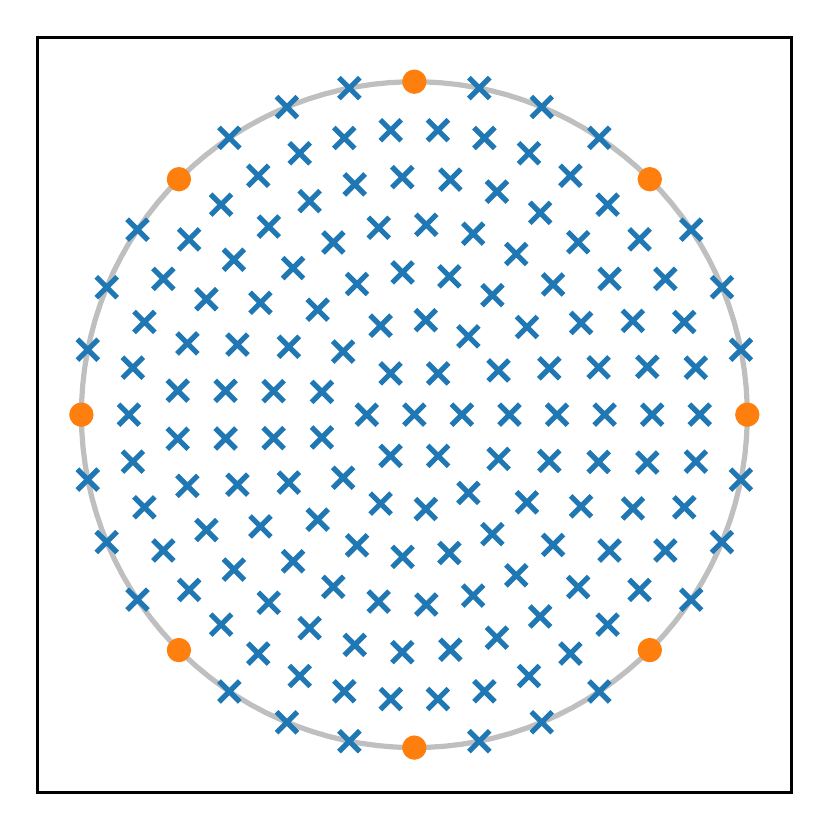}
\caption{$n_A + n_B = 165$}
\end{subfigure}
\begin{subfigure}{0.3\textwidth}
\includegraphics[width = \textwidth]{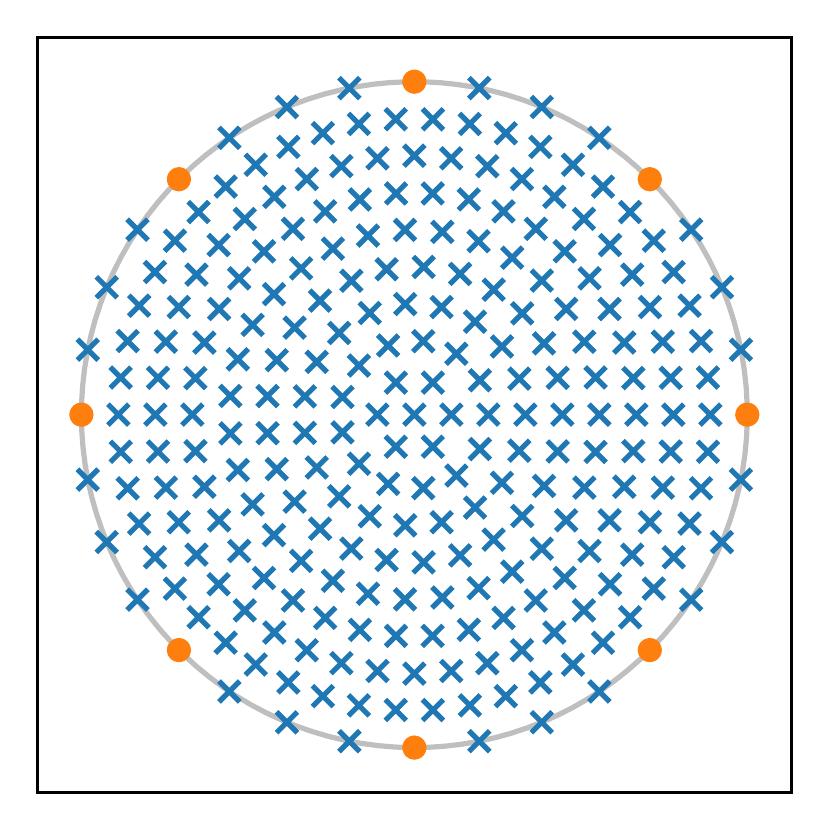}
\caption{$n_A + n_B = 259$}
\end{subfigure}
\begin{subfigure}{0.3\textwidth}
\includegraphics[width = \textwidth]{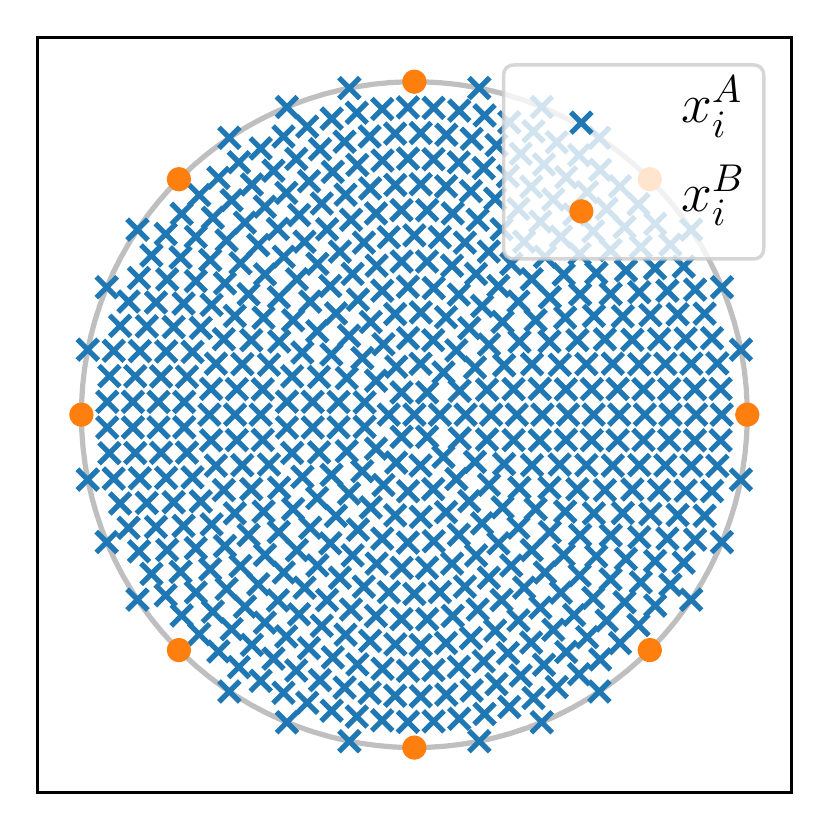}
\caption{$n_A + n_B = 523$}
\end{subfigure}
\caption{Typical sets of collocation points $x_i^A$ and $x_i^B$ that were used to discretise the PDE. }
\label{fig:design}
\end{figure}

For illustration, we first considered simulated data and a coarse collocation method which did \emph{not} model discretisation error.
This was compared to a reference posterior, obtained using a brute-force symmetric collocation forward solver with a large number of collocation points.
(Of course, it is impractical to use a large number of collocation points in the applied context due to the associated computational cost.)
The result, shown in Fig. \ref{ex fail}, was a posterior that did not contain the true data-generating field $a^\dagger$ in its region of support.
This result was observed to be typical for $n_{\mathcal{A}} + n_{\mathcal{B}} < 250$ and motivates the formal uncertainty quantification for discretisation error that is provided by probabilistic numerical methods in this paper.

\begin{figure}[t!]
\centering
\includegraphics[width = 0.5\textwidth]{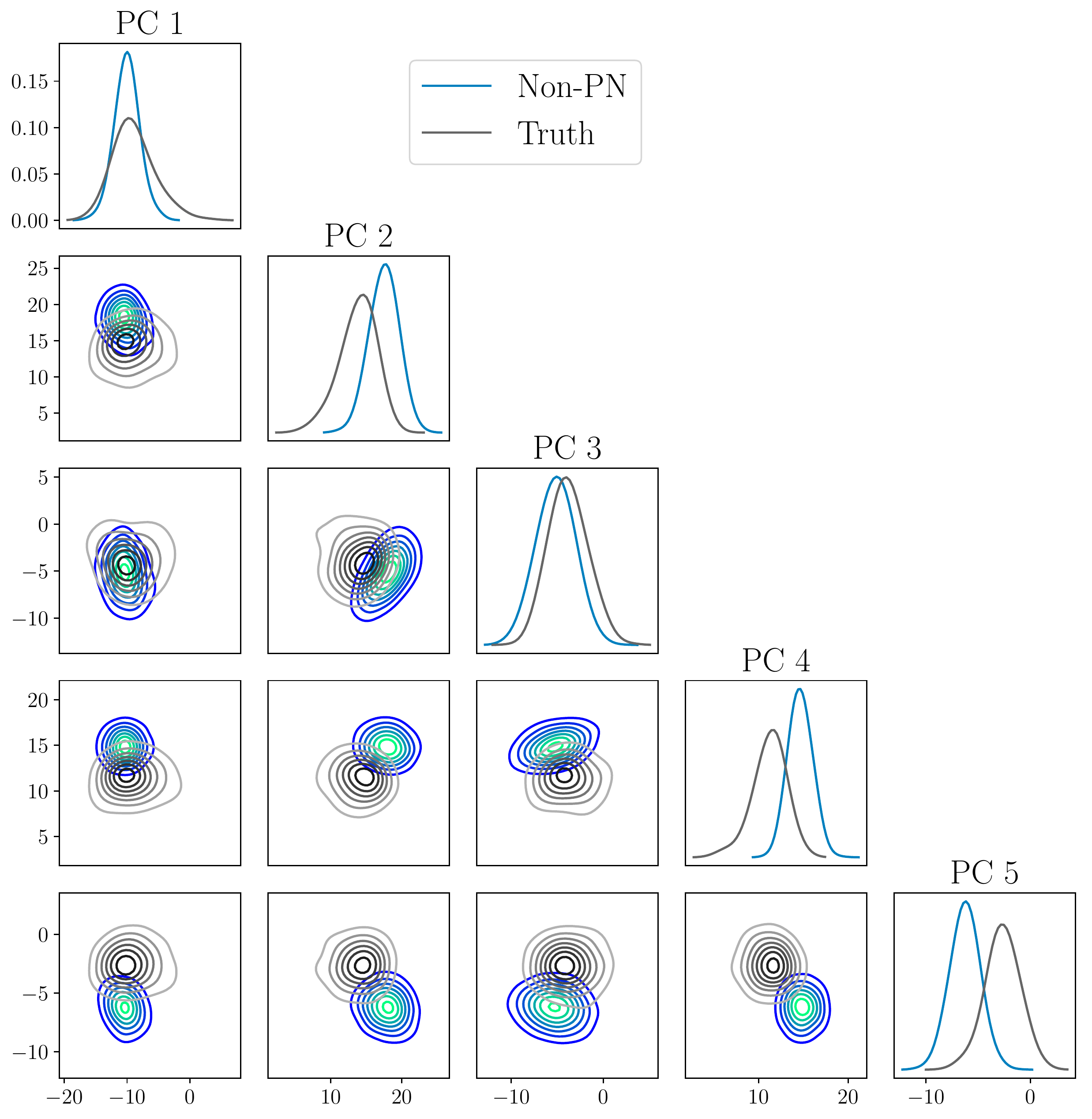}
\caption{Failure case:
Here a small number $n_A+n_B = 165$ of collocation points was used to discretise the PDE, but the uncertainty due to discretisation was not modelled.
The reference posterior distribution (grey) was compared to the approximation to the posterior obtained when the PDE is discretised (blue) and the discretisation error is not modelled (``Non-PN'').
Projections onto principal components (PC) of the reference posterior are displayed.
It is observed that the approximated posterior is highly biased. }
\label{ex fail}
\end{figure}

\begin{figure}[t!]
\centering
\includegraphics[width=\textwidth]{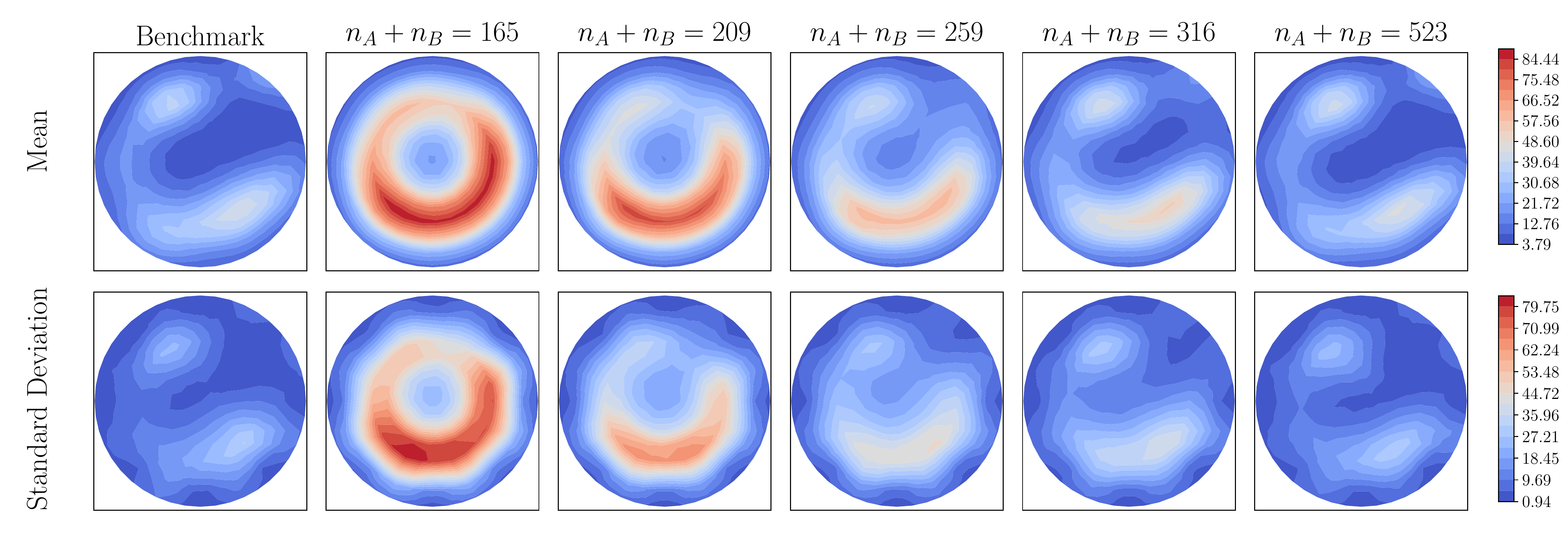}
\caption{Posterior means and standard-deviations for the recovered conductivity field.
The first column represents the reference solution, obtained using a symmetric collocation forward solver with a large number of collocation points.
The remaining columns represent the recovered field when probabilistic numerical methods are used based on $n_A+n_B$ collocation points as illustrated in Fig. \ref{fig:design}. }
\label{fig:static}
\end{figure}

The experimental data used for this assessment were obtained as a single frame (time point 14) from the larger temporal dataset.
In Fig. \ref{fig:static}(a) we show the posterior mean estimate, together with its posterior variance, for a reference conductivity field generated using a high-quality symmetric collocation forward solver with $n_A + n_B = 207$ collocation points.
Adjacent, in Fig. \ref{fig:static} we show the posterior mean and variance for the conductivity field obtained with probabilistic numerical methods for increasing values of $n_A + n_B$.
It is seen that both the posterior mean and posterior standard deviation produced with the probabilistic numerical method converge to the reference posterior as the number of collocation points is increased.
However, at coarse resolution, the posterior variance is inflated to reflect the contribution of an discretisation uncertainty to each numerical solution of the forward problem.
This provides automatic protection against the erroneous results seen in Fig. \ref{ex fail}.

In Fig. \ref{fig:int var} we plot the number $n_A + n_B$ of collocation basis points versus the integrated posterior standard-deviation for the unknown conductivity field.
These results demonstrate the computation-precision trade-off that is made possible with probabilistic numerical methods, and are consistent with the preliminary investigation in \cite{Cockayne2016,Cockayne2016b}.
Next, we turn to the temporal problem that motivates this research.

\begin{figure}[t!]
\centering
\includegraphics[width = 0.49\textwidth]{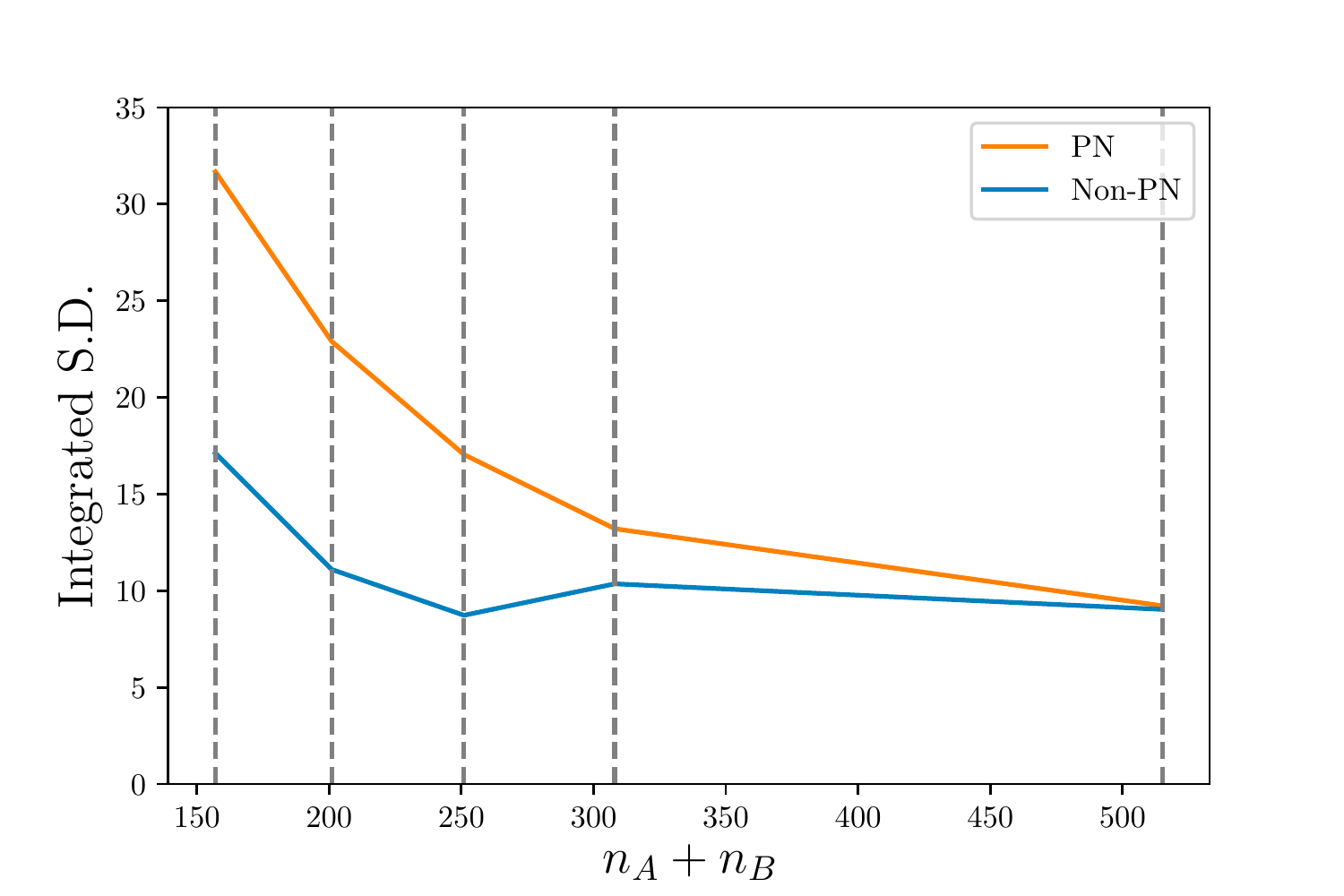}
\caption{Posterior standard-deviation for the conductivity field, integrated over the domain $D$, as a function of the number $n_A+n_B$ of collocation points.
The blue curve represents the standard case where error due to discretisation of the PDE is not quantified (``Non-PN'') whilst the red curve represents the case where a probabilistic numerical method is used to provide uncertainty quantification for the PDE solution itself (``PN''). }
\label{fig:int var}
\end{figure}

\subsubsection{Temporal Recovery Problem}

For this experiment, data were obtained at 49 regular time intervals.
Times 1-10 were obtained before injection of the potassium chloride solution, while the injection occurred rapidly, between frames 10 and 11.
The remaining time points 12-49 capture the diffusion and rotation of the liquids, which is the behaviour that we hope to recover.

The parameter $\lambda$ controls the temporal smoothness of the conductivity field in the prior model.
Three fixed values, $\lambda \in \{10, 100, 1000\}$ were considered in turn, representing decreasing levels of smoothness.
The case of no temporal regularisation was also displayed.
Our method was applied to estimate the time-evolution of the field.
Results are shown in Fig. \ref{fig:temporal}.
The counter-clockwise rotation of the fluid was first clearly seen for $\lambda = 100$, whilst the value $\lambda = 10$ represented too much temporal regularisation, which caused this information to be lost.
On the other hand, the predictive posterior in Eqn. \ref{eq: predictive posterior} is trivial in the limit of large $\lambda$, so that in our context smaller values of $\lambda$ are preferred.
It is expected that an analogous calibration can be performed in the real-world context.

\begin{figure}[t!]
\centering

\includegraphics[width=\textwidth]{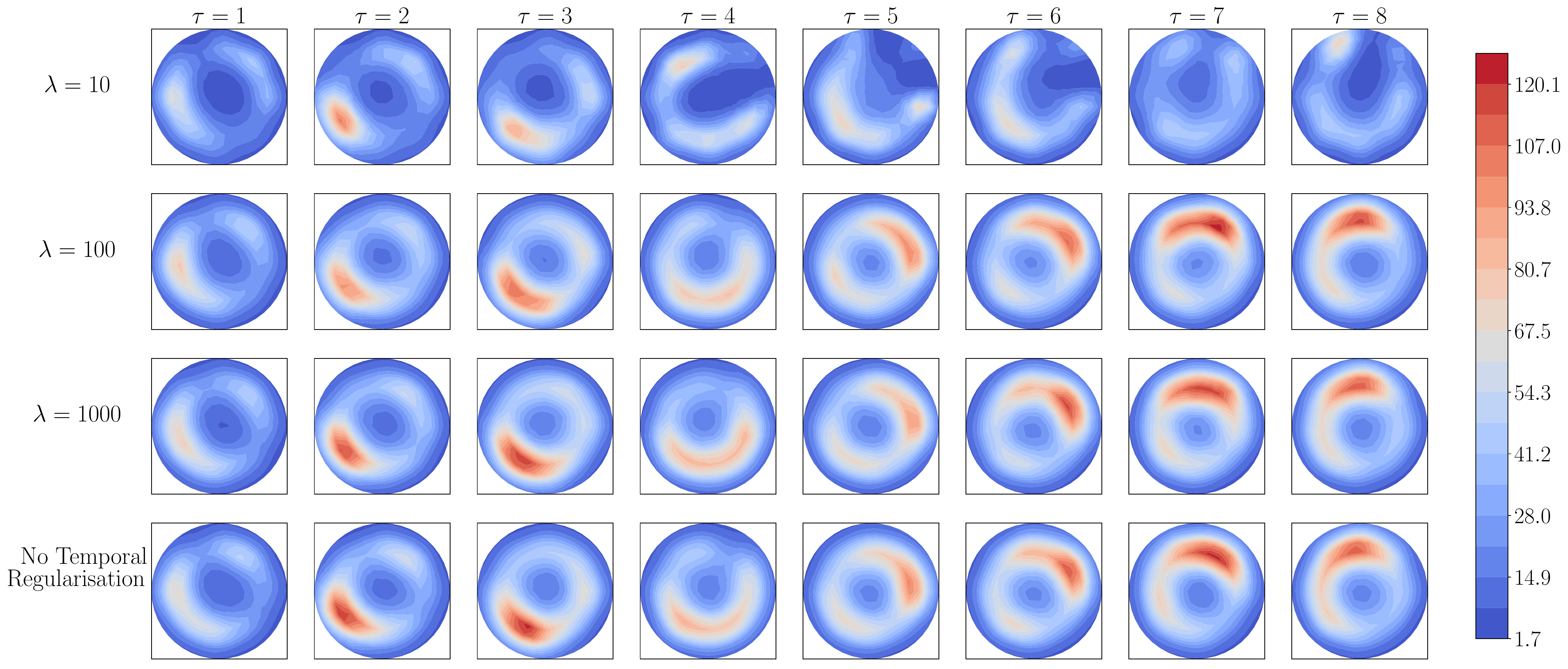}
\caption{Posterior mean for the conductivity field $a(\cdot,t)$, shown as a function of the time index.
Here we consider the dependence of the recovered field on the choice of the temporal covariance parameter $\lambda$.
The value $n_A + n_B = 209$ was used. }
\label{fig:temporal}
\end{figure}

To assess whether the problems of bias and over-confidence due to discretisation can be mitigated in the temporal context, where discretisation errors are propagated and accumulate over time, we fixed $\lambda = 100$ and inspected the posterior over the coefficients $\psi_i$ at the final time point $t_n$.
Results in Fig. \ref{temporal coefficients} confirmed that the posterior $\Pi_n^*$ (red) was inflated relative to the standard approximate posterior (blue) and tended to cover more of the true posterior $\Pi_n$ (grey) in its effective support.
This provides empirical evidence to support the use of the proposed posterior $\Pi_n^*$.

\begin{figure}[t!]
\centering
\includegraphics[width = 0.5\textwidth]{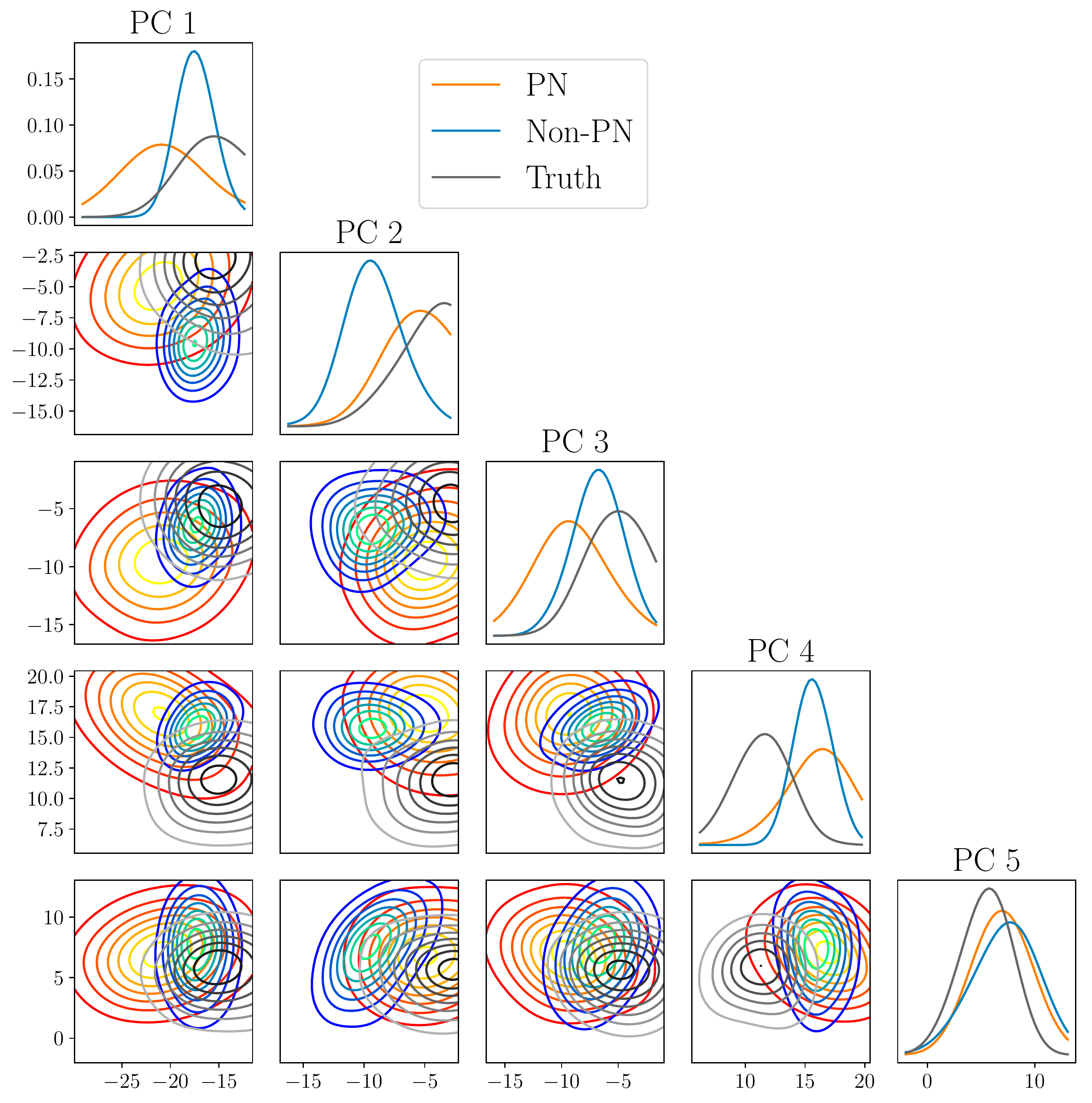}
\caption{Posterior distribution over the coefficients $\psi_i$, at the final time point $t_n$.
Here a small number $n_A+n_B = 165$ of collocation points was used to discretise the PDE.
The reference posterior distribution (grey) was compared to the approximation to the posterior obtained when discretisation of the PDE is not modelled (``Non-PN'') and modelled (``PN''). 
Projections onto principal components (PC) of the reference posterior are displayed. }
\label{temporal coefficients}
\end{figure}

In Fig. \ref{fig:int var 2} we again plot the number $n_A + n_B$ of collocation basis points versus the integrated posterior standard-deviation for the unknown conductivity field, again at the final time point.
These results demonstrate a computation-precision trade-off similar to that which was observed for the static recovery problem.
Compared to the static recovery problem in Fig. \ref{fig:int var}, however, we observed greater inflation of the posterior standard deviation when probabilistic numerical methods were used.
This reflects the fact that we have constructed a full probability model for the effect of discretisation error, which is able to capture how these errors propagate and accumulate within the computational output.

\begin{figure}[t!]
\centering
\begin{subfigure}{0.49\textwidth}
\includegraphics[width = \textwidth]{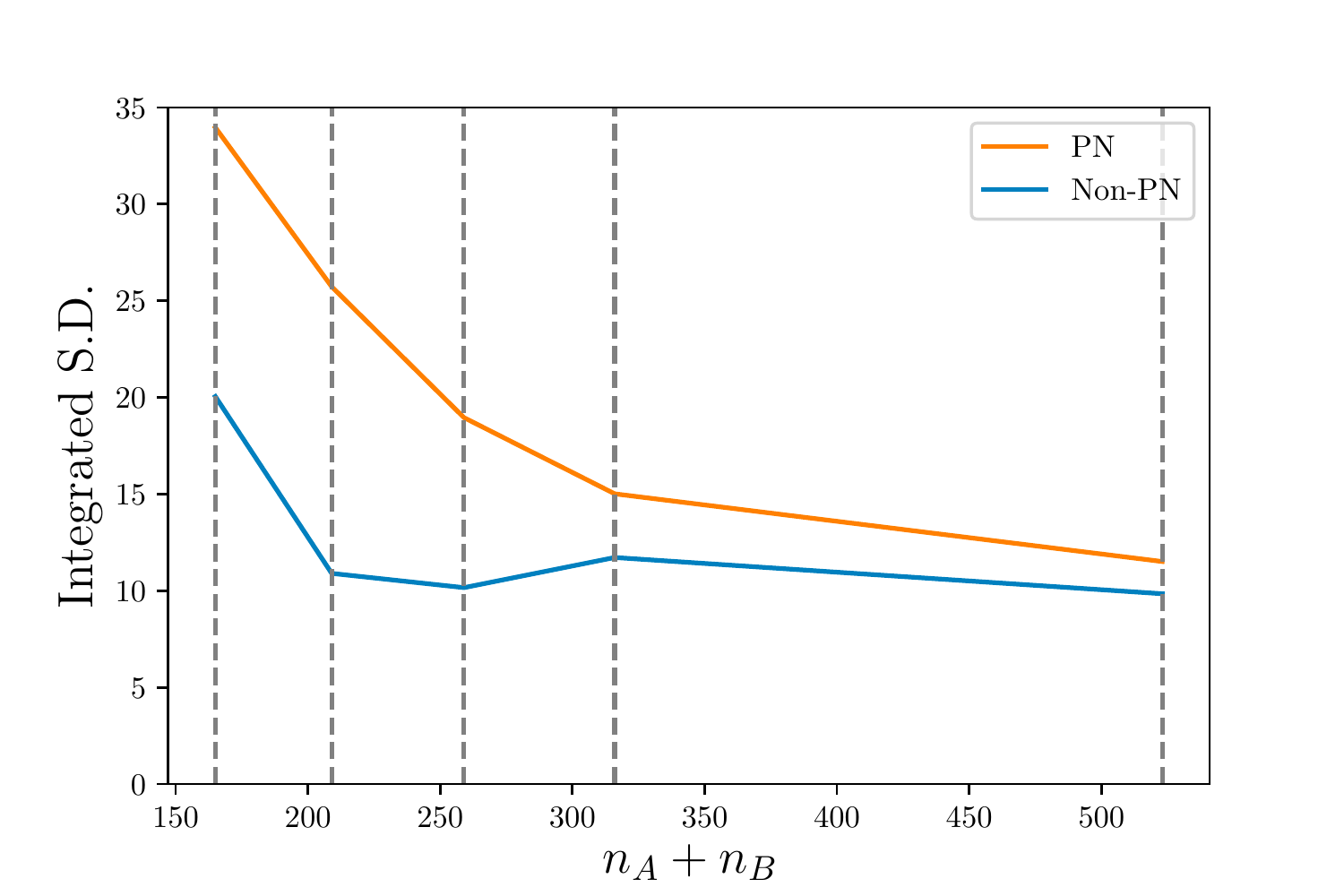}
\caption{}
\label{tmp fig 1}
\end{subfigure}
\begin{subfigure}{0.49\textwidth}
\includegraphics[width = \textwidth]{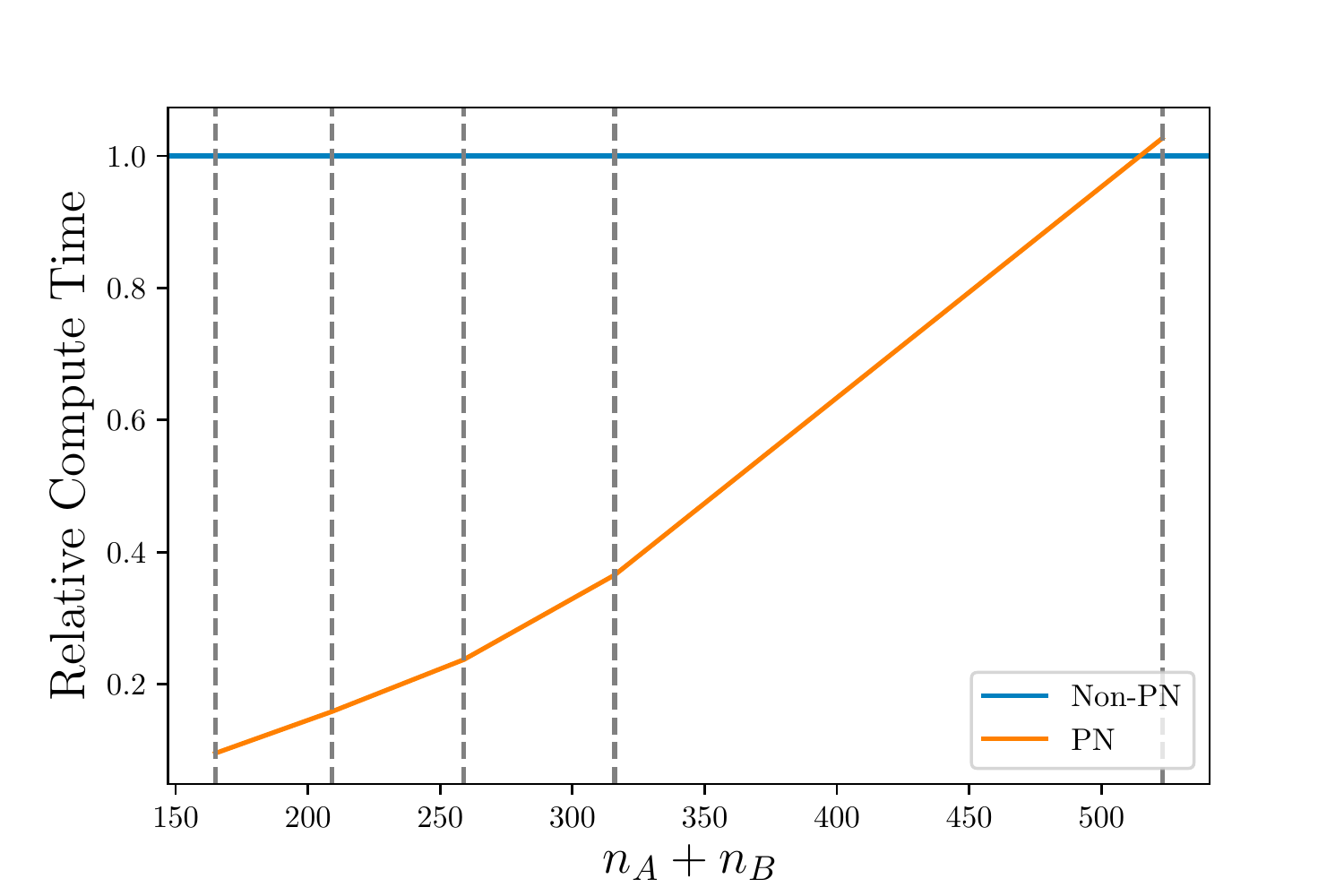}
\caption{}
\label{tmp fig 2}
\end{subfigure}
\caption{Temporal recovery problem: (a) Posterior standard-deviation for the conductivity field $a(\cdot,t_n)$ at the final time point $t_n$, integrated over the domain $D$.
(b) Computational time required by the proposed method, relative to the (non-probabilistic) symmetric collocation method applied on a resolved grid.}
\label{fig:int var 2}
\end{figure}

\section{Discussion} \label{sec:discussion}

The motivation for this research was industrial process monitoring, but the associated methodological development was general.
In particular, we addressed the important topic of how to perform uncertainty quantification for numerical error due to discretisation of the physical governing equations specified through a PDE.
Typically this source of error is ignored, or its contribution bounded through detailed numerical analysis, such as \cite{Schwab2012}.
However, in the temporal setting, theoretical bounds are difficult to obtain due to propagation and accumulation of errors, so that it is unclear how to proceed.

In this work we proposed a statistical solution, wherein a probabilistic numerical method was used to provide uncertainty quantification for the discretisation error associated with a collocation-type numerical method.
Aided by sequential Monte Carlo sampling methods, it was shown how this model for \emph{discretisation uncertainty} can be employed in the temporal context.
The result was a more comprehensive quantification of uncertainty, that accounts for both statistical uncertainty and for propagation and accumulation of discretisation uncertainty in the final output.
For our motivating industrial application, this work is expected to facilitate more reliable anticipation and pro-active control of the hydrocyclone, to ensure safety in operation \citep{Bradley2013}.
Beyond that, it is anticipated that the mitigation of bias and over-confidence observed in our experimental results is a feature of probabilistic numerical methods in general.

The methods that we pursued differ in a fundamental sense to techniques that seek to \emph{emulate} the forward model.
Emulation, as well as dimension reduction methods, have been widely used in static recovery problems to reduce the computational cost of repeatedly solving the governing PDE \citep{Marzouk2007,Marzouk2009,Cotter2010,Schwab2012,Cui2016,Chen2016,Chen2016b}.
Notably \cite{Stuart2016,Calvetti2017} considered integrating the emulator uncertainty into inference for model parameters. 
However, to train an emulator it is usually required to have access to a training set of parameters $a$ for which the exact solution $u$ of the PDE is provided.
Thus the focus of emulation is related to generalisation in the $a$ domain, as opposed to quantification of discretisation uncertainty in the $u$ domain.
An interesting extension of this work would be to combine these two complementary techniques; this would be expected to reduced the computational cost of the proposed method.

The principal limitation of our approach was that a Markov temporal evolution of the conductivity field $a(\cdot,t)$ was assumed.
Physical consideration suggest that the Markov assumption is incorrect, since time-derivatives of all orders of this field will vary continuously and thus encode information that is useful.
However, it is not clear how this information can be encoded into a prior for the temporal evolution of the conductivity field whilst preserving the computationally convenient filtering framework.
On the other hand, the Markov prior can be statistically justified in that it represents an encoding of partial prior information \citep{Potter1983}.
It remains a problem for future work to investigate the potential loss of estimation and predictive efficiency as a result of encoding only partial information into the prior model.

The second limitation we highlight is that the prior model for the potential field $u(\cdot,t)$ did not include a temporal component.
This choice was algorithmically convenient, as it de-coupled each of the forward problems of solving the PDE, such that each time a probabilistic numerical method was called, it could be implemented ``out of the box''.
Nevertheless, a temporal covariance structure in the parameter $a$ implies that there also exists such structure in $u$ and the effect of not encoding this aspect of prior information should be further investigated.

Overall, we are excited by the prospect of new and more powerful methods for uncertainty quantification that can deal with both statistical and discretisation error in a unified analytical framework.

\paragraph{Acknowledgements:}
The authors are grateful for detailed suggestions from the Associate Editor and two anonymous Reviewers.
\if1\blind
{
This research was supported by the Australian Research Council Centre of Excellence for Mathematical and Statistical Frontiers and by the Key Technology Partnership program at the University of Technology Sydney.
CJO and MG were supported by the Lloyd's Register Foundation programme on data-centric engineering at the Alan Turing Institute, UK.
MG was supported by the EPSRC grants [EP/K034154/1, EP/R018413/1, EP/P020720/1, EP/L014165/1], an EPSRC Established Career Fellowship [EP/J016934/1] and a Royal Academy of Engineering Research Chair in Data Centric Engineering.
The collection of tomographic data was supported by an EPSRC grant [GR/R22148/01].
This material was based upon work partially supported by the National Science Foundation under Grant DMS-1127914 to the Statistical and Applied Mathematical Sciences Institute. Any opinions, findings, and conclusions or recommendations expressed in this material are those of the author(s) and do not necessarily reflect the views of the National Science Foundation.

For the numerical results reported in Section~\ref{sec:experiments} we thank T.~J.~Sullivan for the use of computing facilities at the Freie Universit\"at Berlin, funded by the Excellence Initiative of the German Research Foundation.
} \fi

\appendix

\section{Proof of Results in the Main Text} \label{app: proofs}

\begin{proof}[Proof of Proposition \ref{prop: is cts}]
Let $L_H^2 = \{v : \bar{D} \times \Omega \rightarrow \mathbb{R} \text{ s.t. } \mathbb{E}\|v\|_H^2 < \infty\}$, which is a Banach space with norm $(\mathbb{E}\|\cdot\|_H^2)^{1/2}$; see section 2.4 of \cite{Dashti2013}.
Following Thm. 2.10 in \cite{Dashti2013}, consider the partial sums
$$
\log \; a^N(\cdot,t) = \sum_{i=1}^N i^{- \alpha} \psi_i(t) \phi_i(\cdot) .
$$
For $N > M$ we have
\begin{eqnarray*}
\mathbb{E} \| \log \; a^N(\cdot,t) - \log \; a^M(\cdot,t) \|_{H}^2 & = & \mathbb{E} \sum_{i=M+1}^N i^{-2\alpha} |\psi_i(t)|^2 \\
& \leq & \sum_{i=M+1}^N i^{-2\alpha} [(m_\psi^{\max})^2 + k_\psi^{\max}] \\
& \leq & [(m_\psi^{\max})^2 + k_\psi^{\max}] \sum_{i=M+1}^\infty i^{-2 \alpha} .
\end{eqnarray*}
Since $\alpha > 1/2$, the RHS vanishes as $M \rightarrow \infty$.
Thus, as $L_H^2$ is a Banach space, $\log \; a(\cdot,t)$ exists as an $L_H^2$ limit.
It follows that $\log \; a(\cdot,t)$, and hence $a(\cdot,t)$, takes values almost surely in $C^1(\bar{D})$.
\end{proof}

\begin{proof}[Proof of Corollary \ref{field increment corr}]
From direct algebra:
\begin{eqnarray*}
\theta_\Delta(x) \; = \; \theta(x,t+s) - \theta(x,t) & = & \sum_{i=1}^\infty i^{-\alpha} [\psi_i(t+s) - \psi_i(t)] \phi_i(x) \\
& = & \sqrt{ \lambda (s + \tau) } \sum_{i=1}^\infty i^{-\alpha} \xi_i \phi_i(x)
\end{eqnarray*}
where the $\xi_i$ are independent $\mathrm{N}(0,1)$.
From the Karhounen-Lo\'{e}ve theorem \citep{Loeve1977}, this is recognised as a Gaussian random field with mean function $m_\Delta(x) = 0$ and covariance function 
\begin{eqnarray*}
k_\Delta(x,x') & = & \lambda (s + \tau) \sum_{i=1}^\infty i^{-2\alpha} \phi_i(x) \phi_i(x')  \\
& = & \lambda (s + \tau) k_\phi(x,x') , 
\end{eqnarray*}
as claimed.
\end{proof}

\begin{proof}[Proof of Proposition \ref{prior well defined for u}]
Let $H(k)$ denote the reproducing kernel Hilbert space associated with a kernel $k$.
In \cite{Cialenco2012}, Lemma 2.2, it was established that a generic integral-type kernel $k_u$, as in Eqn. \ref{int kernel}, corresponds to the covariance function for a Gaussian process that takes values almost surely in $H(k_u^0)$.
To complete the proof, Corr. 4.36 (p131) in \cite{Steinwart2008} establishes that if $k_u^0 \in C^{\beta \times \beta}(\bar{D} \times \bar{D})$ then $H(k_u^0) \subset C^{\beta}(\bar{D})$.
\end{proof}

\begin{proof}[Proof of Proposition \ref{fill dist result}]
Let $\mu(x)$ and $\sigma(x)$ denote, respectively, the posterior mean and standard deviation of $u(x)$ under the probabilistic meshless method.
Prop. 4.1 of \cite{Cockayne2016} established that the posterior mean $\mu(x)$ satisfies $|\mu(x) - u(x)| \leq \sigma(x) \|u\|_{H(k_u)}$ and Prop. 4.2 of \cite{Cockayne2016} established that the posterior standard deviation $\sigma(x)$ satisfies $\sigma(x) \leq C h^{\beta - 2 - d/2}$ for some constant $C$ independent of $x \in D$.
In particular we have $\|\mu_j - \mathcal{P}u\|_2 = O(h^{\beta - 2 -d/2})$.
Lastly, Thm. 4.3 of \cite{Cockayne2016} established a generic rate of contraction for the mass of a Gaussian distribution of $O(\epsilon^{-2} h^{2\beta - 4 - d})$, as required.
(Note that these results are specific consequences of more general results found in Lem. 3.4 in \cite{Cialenco2012} and Secs. 11.3 and 16.3 of \cite{Wendland2005}.)
\end{proof}

\section{Details of the Markov Kernel Used} \label{ap: markov kernels}

This appendix contains a description of the Markov transition kernel that was used.
Indeed, for the Markov transition kernel $M_{n-1}$ used in the {\bf Move} step in Section \ref{subsec:SMC}, we employed the pre-conditioned Crank--Nicholson method.
This will now be described.

Let $\Pi$ be a probability distribution on a measurable space $(\Theta,\mathcal{B})$, such that the Radon-Nikodym derivative $\mathrm{d}\Pi / \mathrm{d}\Pi_0$ is well-defined for a fixed reference distribution $\Pi_0$.
Recall that a Markov transition kernel $M$ which leaves $\Pi$ invariant is a function $M : \Theta \times \mathcal{B} \rightarrow [0,1]$ such that
\begin{enumerate}
\item the map $\theta \mapsto M(\theta,B)$ is $\mathcal{B}$-measurable for all $B \in \mathcal{B}$
\item the map $B \mapsto M(\theta,B)$ is a probability measure on $(\Theta,\mathcal{B})$ for all $\theta \in \Theta$
\item invariance; $\Pi(B) = \int M(\theta,B) \mathrm{d}\Pi(\theta)$ for all $B \in \mathcal{B}$.
\end{enumerate}
The pre-conditioned Crank--Nicholson method (with step size $\beta \in (0,1)$)
\begin{eqnarray*}
\theta^* & := & \sqrt{(1 - \beta^2)} \theta + \beta \xi, \qquad \xi \; \sim \; \Pi_0 \\
\theta_{\text{new}} & = & \left\{ \begin{array}{ll} \theta^* & \text{with probability } \alpha(\theta,\theta^*) = \min\left\{ 1 , \frac{\mathrm{d}\Pi}{\mathrm{d}\Pi_0}(\theta^*) / \frac{\mathrm{d}\Pi}{\mathrm{d}\Pi_0}(\theta) \right\} \\ \theta & \text{otherwise} \end{array} \right.
\end{eqnarray*}
for generating the next state $\theta_{\text{new}}$ of the Markov chain, given the current state is $\theta$, corresponds to a Markov transition kernel
\begin{eqnarray*}
M(\theta,B) = \int 1[\theta^* \in B] \alpha(\theta,\theta^*) + 1[\theta \in B] (1 - \alpha(\theta,\theta^*)) \mathrm{d}\Pi_0(\xi)
\end{eqnarray*}
that leaves $\Pi$ invariant.
The associated Markov chain has been shown to have non-vanishing acceptance probability when $\Theta$ is a Hilbert space and $\Pi_0$ is a Gaussian distribution \citep[Thm. 6.4 of][]{Cotter2013}.
This was the Markov transition kernel that we employed, with $\beta$ tuned to achieve an acceptance rate of between 10\%--25\%, $\Theta$ being the state space of $\theta_n$  and $\Pi_0$ being the prior $\Pi_{0,1}$, defined in the main text.
Nevertheless, it is not the unique Markov transition kernel that could be used; see \cite{Cotter2013} for several examples of Markov transition kernels that are well-defined in the Hilbert space context.


\end{document}